\documentclass[psamsfonts]{amsart}

\usepackage{amssymb}
\usepackage{amsmath}
\usepackage{amsthm}
\usepackage{mathrsfs}
\usepackage{graphicx}
\usepackage{bbm}
\usepackage{psfrag}
\usepackage{hyperref}
\usepackage{comment}

\theoremstyle{plain}
\newtheorem{thm}{Theorem}[section]
\newtheorem{lemma}[thm]{Lemma}

\theoremstyle{definition}

\numberwithin{equation}{section}

\begin{document}

\title{Executing large orders in a microscopic market model}
\author{Alexander Wei\ss{}}
\address{Weierstrass Institute for Applied Analysis and Stochastics\\
		 Mohrenstrasse 39\\
		 10117 Berlin\\
		 Germany
}
\email{alexander.weiss@wias-berlin.de}
\thanks{Work supported by the DFG Research Center MATHEON}
\date{\today}
\keywords{market micro structure, illiquid markets, optimal trading strategies}
\subjclass[2000]{91B24 62P05}
\begin{abstract}
In a recent paper, Alfonsi, Fruth and Schied (AFS) propose a simple order book based model for the impact of large orders on stock prices. They use this model to derive optimal strategies for the execution of large orders. We apply these strategies to an agent-based stochastic order book model that was recently proposed by Bovier, \v{C}ern\'{y} and Hryniv, but already the calibration fails. In particular, from our simulations the recovery speed of the market after a large order is clearly dependent on the order size, whereas the AFS model assumes a constant speed. For this reason, we propose a generalization of the AFS model, the GAFS model, that incorporates this dependency, and prove the optimal investment strategies. As a corollary, we find that we can derive the ``correct'' constant resilience speed for the AFS model from the GAFS model such that the optimal strategies of the AFS and the GAFS model coincide. Finally, we show that the costs of applying the optimal strategies of the GAFS model to the artificial market environment still differ significantly from the model predictions, indicating that even the improved model does not capture all of the relevant details of a real market.
\end{abstract}

\maketitle
\section{Introduction}
For a long time, financial mathematics mainly focused on asset pricing, but the scope has been extended in the last years. One of the current topics of interest is the theory of optimal trading strategies for the execution of large orders. Here, a trader would like to purchase\footnote{In this article, we focus on a trader {\it purchasing} shares, since all models mentioned here either consider only this part of the problem or work symmetrically for buyers and sellers.} a huge volume of shares up to time $T$. Since the supply of limit orders for a certain price is limited, the trader will not be able to trade the whole order for the current price, but he or she will suffer from an adverse price movement. This additional price impact, induced by the trader's own trading, can be lessened if he or she gives the market time to recover; the {\it best price} returns to previous levels. However, the time interval $[0,T]$ is assumed to be too short in order to wait for a full recovery of the market. The {\it optimal execution problem} asks for the optimal splitting and the optimal trading times to minimize the expected price impact.

There have been several models to solve the optimal execution problem, motivated by empirical findings (for references see next paragraph); yet, since we do not know if these models capture all relevant features of real markets, we cannot be sure that the strategies work in reality, and tests on real markets would be an expensive experiment. For this reason, {\it microscopic market models} are an excellent tool for testing theoretical models of optimal trading strategies. Based on assumptions about the market participants' behavior, these models simulate the trading of financial assets on the level of single traders or orders \cite{giardina03,smith03,bovier06}. The emerging price processes show typical features of real markets \cite{cont01,potters03,almgren05}. Hence, microscopic models provide artificial, yet reasonable, market environments that allow for applying optimal trading strategies without costs or risk, comparing the numerical results with the theoretical expectations and resolving deviations by an improvement of the underlying market assumptions with respect to the empirical findings. In this paper, we exemplify how this approach can improve a solution for the optimal execution problem.

All approaches to the optimal execution problem rely on two empirical findings that have been validated in many studies (see \cite{schoeneborn08}, pp. 3, for a list of references): First, a large order has an impact on its price; second, this impact decreases in time, but it does not vanish completely. That implies the costs of all subsequent orders are influenced by the impact of a large order. These two effects are called {\it temporary} and {\it permanent impact}. Many models implement these observations straightly \cite{bertsimas98,almgren01,huberman05}: They consider a stochastic process that simulates the current best price evolving independently from the large trader's action in time, and two functions mapping the volume of a large order to the temporary or, respectively, permanent impact. When a large order is executed, the corresponding impacts are just added to the price. Yet, it is doubtful if the complex dynamics of limit order books ({\it LOB}), which underlie most modern markets, can be captured by looking at the best price only. Therefore, recent models attempt to take the dynamics of the whole order book into account. Obizhaeva and Wang introduced a model with an underlying block shaped LOB and calculated the optimal trading strategy in terms of a recursive formula by applying Bellman equations \cite{obizhaeva05}. Alfonsi, Fruth and Schied introduced a generalization of this model for general order book shapes and gave an explicit solution for the optimal trading strategy with respect to their market model (introduced in \cite{alfonsi09} and revisited in \cite{alfonsi09b}); this model is the one we will test in a microscopic market environment, and we refer to it as the {\it AFS model}. 

The AFS model describes the underlying market by two parameters: The shape of the (continuous) LOB given by a {\it shape function} $f$ and a positive constant $\rho$ expressing the {\it resilience speed} of the order book. There are two versions of the model: In the first one, the consumed volume recovers exponentially fast; in the second version, the best price recovers in this way. The shape of the order book is {\it static} such that there is a bijection between the impact on the best price and on the volume. Thus, the response of the order book to the execution of a large order depends on the current price impact only, but not on possible executions before.

To test the optimal AFS strategies, we need to select a microscopic market model. The model that serves best as virtual market environment was introduced by Bovier, \v{C}erny and Hryniv and is called the {\it Opinion Game} \cite{bovier06}. It simulates a family of traders on the level of a {\it generalized order book}. The underlying idea of the generalized order book is that every trader has an individual, subjective opinion about the current {\it fair} price. Instead of orders, the generalized order book records these opinions. Thus, it also captures traders who are willing to trade for a price close to the best quotes but have not placed public orders (in some markets it is also possible to place hidden or partially visible orders \cite{frey09, bessembinder09}). These traders offer {\it hidden liquidity}; they will influence the price impact when an order is executed but do not appear in the order book \cite{weber05}. Thus, the Opinion Game provides a more realistic market response to orders than classical order book models.

In order to apply the AFS strategies in the Opinion Game, we have to determine the correct values for $f$ and $\rho$. There are several problems to find the value for $\rho$. First, the AFS model does not assume a permanent impact; second, the market recovery is only poorly approximated by an exponential function; third, $\rho$ does not exist as a constant value but depends on the traded volume. While the first two items can be bypassed, the third item strongly conflicts with the assumptions of the AFS model. For this reason we introduce a generalization of the AFS model that we call the {\it generalized AFS model} or {\it GAFS model}. The GAFS model substitutes $\rho$ by $\bar{\rho}$ that is a function of an order's price impact or volume impact, depending on the model version. Furthermore, we extend the results of the AFS model by proving that there exists a unique, deterministic optimal trading strategy for the GAFS model. It turns out that, although $\bar{\rho}$ is a function, the optimal strategy evaluates it for one value only. Consequently, the optimal strategies of the AFS and the GAFS models coincide when $\rho$ is chosen to be this value. In this sense, the AFS model is also sufficient for the order impact dependent case, but the GAFS model is needed to calibrate it correctly.

After calibrating the (G)AFS model to the Opinion Game, we calculate the optimal strategies for several parameter sets, apply these strategies to the Opinion Game, and sample their impact costs. On a general level, the sampled costs show the expected {\it natural} behavior; for instance, the costs decrease if the available trading time $T$ or the number of trading opportunities within $[0,T]$ become larger. Furthermore, the simulations reinforce the advantages of the GAFS model compared to the AFS model. We show that the AFS model performs worse than the GAFS model for a {\it bad}, yet reasonable, choice of the value for $\rho$. On the other hand, we find that, in comparison to the predicted costs, the sampled costs of the GAFS strategies are up to four times higher. This shows that the (G)AFS model does not capture all relevant details of the Opinion Game's order book dynamics, indicating that the optimal (G)AFS strategies could also perform worse than theoretically expected on real markets.

In Section \ref{sec:AFSmodel}, we introduce the AFS model and restate its optimal trading strategies. In Section \ref{sec:opinionGame}, we present this version of the Opinion Game that we used to analyze the AFS model. In Section \ref{sec:parameters}, we determine $f$ and $\rho$ in the Opinion Game, which leads to the GAFS model. Finally, in Section \ref{sec:numerics}, we apply the GAFS optimal strategies in the Opinion Game, and compare the resulting costs for several parameter sets. Furthermore, we show that the GAFS strategies perform better than the AFS strategies with an suboptimal choice of $\rho$ in the Opinion Game.

\section{The market model of Alfonsi, Fruth and Schied and its optimal execution strategies}\label{sec:AFSmodel}
A trader would like to purchase $X_0>0$ shares within a time period $[0,T]$, $T>0$. $X_0$ is assumed to be large such that the trader's order has an impact on the price and the underlying limit order book. We will refer to this trader as {\it large trader} in the following. Because we consider a buy order, we first define how the {\it upper part} of the LOB, which contains the sell limit orders, is modeled. As long as the large trader does not take action, the LOB is described by the {\it unaffected best ask price} $A^0:=(A^0_t)_{t\geq0}$ and by a {\it shape function} $f:\mathbb{R}\to (0,\infty)$ (see Figure \ref{fig:AFSLob}). $A^0$ is a martingale on a given filtered probability space $(\Omega, (\mathscr{F}_t)_{t\geq0},\mathscr{F},P)$ satisfying $A^0_0=A_0$ for some $A_0\in\mathbb{R}$; $f$ is a continuous function. The amount of shares available for a price $A^0_t+x$, $x\geq 0$, at time $t$ is then given by $f(x)dx$. Notice that the shape of the order book with respect to the best ask price is static.

\begin{figure}
\includegraphics[width=.9\textwidth]{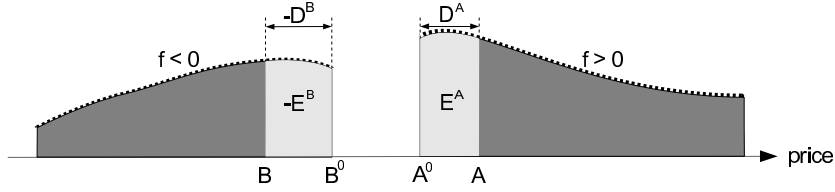}
\caption{The order book of the AFS model. For simplicity, we have left off the time index $t$.}
\label{fig:AFSLob}
\end{figure}

Now, assume the large trader acts for the first time and purchases $x_0$ shares at time $t_0$; he or she consumes all shares between $A^0_{t_0}$ and $A^0_{t_0}+D^A_{t_0+}$, $D^A_{t_0+}$ being uniquely determined by
\begin{equation}
\int_0^{D^A_{t_0+}}f(x)dx=x_0.
\end{equation}
$D^A:=(D^A)_{t\geq 0}$ is called the {\it extra spread} caused by the large trader. In general, if we know $D^A_{t_n}$ for a trading time $t_n$, $D^A_{t_n+}$ is given by
\begin{equation}\label{eq:dtrade}
\int_{D^A_{t_n}}^{D^A_{t_n+}}f(x)dx=x_n
\end{equation}
whereby $x_n$ is the amount of shares traded at time $t_n$. The large trader is inactive between two trading times, $t_n$ and $t_{n+1}$, and the extra spread recovers. For the exact way of recovery there are two versions considered. To conform to the notation of \cite{alfonsi09}, we first state {\it Version 2}. In this case, $D^A_t$ is defined for $t\in (t_n,t_{n+1}]$ by
\begin{equation}\label{eq:eq2mod2}
D^A_t := e^{-\rho(t-t_n)}D^A_{t_n+}.
\end{equation}
The parameter $\rho$ is a positive constant called the {\it resilience speed}. To complete the definition, we set $D^A_t :=0$ for $t\leq t_0$. Now, we can introduce the {\it best ask price} $A := (A_t)_{t\geq 0}$ by
\begin{equation}
A_t := A^0_t + D^A_t.
\end{equation}
In contrast to $A^0$, $A$ includes the large trader's impact. In particular, the amount of shares available for a price $A^0_t+x$ at time $t$ is given by
\begin{equation}
\left\{\begin{array}{ll}
f(x)dx&\textrm{for }x\geq A_t-A^0_t\\
0&\textrm{otherwise}
\end{array}\right. .
\end{equation}
In other words, every trader in the market experiences the large trader's impact after time $t_0$.

The {\it price impact} $D^A$ can also be expressed in terms of the {\it impact on the volume} $E^A:=(E^A_t)_{t\geq 0}$. Because the shape function $f$ is strictly positive, there is a one-to-one relation between $E^A$ and $D^A$. Given $D^A$, the process $E^A$ is defined by
\begin{equation}
E^A_t := \int_0^{D^A_t}f(x)dx.
\end{equation}
We can generally introduce the antiderivative of $f$,
\begin{equation}
F(x) := \int_0^x f(x)dx,
\end{equation}
to get the relations
\begin{equation}\label{eq:relationDE}
\begin{array}{ccc}
E^A_t = F(D^A_t)&\textrm{and}&D^A_t=F^{-1}(E^A_t).
\end{array}
\end{equation}
By (\ref{eq:dtrade}) and (\ref{eq:relationDE}), we easily conclude
\begin{equation}\label{eq:etrade}
E^A_{t_n+} = E^A_{t_n}+x_n.
\end{equation}
This motivates to define {\it Version 1}, in which we first define $E^A$ and then derive $D^A$ by relation (\ref{eq:relationDE}). We set $E^A_t := 0$ for $t\in [0,t_0]$ and
\begin{equation}\label{eq:eq2mod1}
E^A_t := e^{-\rho(t-t_n)}E^A_{t_n+},\ t\in(t_n,t_{n+1}].
\end{equation}
The equations (\ref{eq:etrade}) and (\ref{eq:eq2mod1}) define $E^A$ completely.

Summarizing, we have introduced two versions of the AFS model: In Version 1, we define the volume impact $E^A$ and assume that it recovers exponentially fast between the large trader's orders. $D^A$ is then derived from $E^A$ by relation (\ref{eq:relationDE}); in Version 2, we first define the price impact $D^A$, assume an exponentially fast recovery and derive $E^A$ from it. Observe that the AFS model recovers completely as the time tends to infinity (and if no more large orders are executed after some finite time); there is no permanent impact. Furthermore, the assumption of an exponential decay of the price impact is under discussion in the scientific community. An empirical study in \cite{bouchaud04} suggests a power-law decay. On a theoretic level, Gatheral proved for a market model similar to the AFS model that an exponential decay can easily imply arbitrage opportunities while power-law decays do not show this undesired property \cite{gatheral09}. Alfonsi and Schied, however, were able to show in \cite{alfonsi09b} that, despite of the similarity to Gatheral's model, the AFS model does not give arbitrage opportunities (under mild assumptions concerning the shape of the order book). The final answer to the question how the decay of the price impact is modeled best has not been given yet. As far as we know, the same question for the volume impact has not been treated.

We cannot exclude a priori that it is reasonable to sell shares and to buy them back later. Thus, we also have to model the impact of (large) sell orders on the LOB. Such orders will be written as orders with negative sign. Let $B^0=(B^0_t)_{t\geq 0}$ be the {\it unaffected best bid price} with
\begin{equation}
B^0_t\leq A^0_t\textrm{ for all }t\geq 0
\end{equation}
as only constraint for its dynamics. The {\it lower part} of the LOB is modeled by the shape function $f$ on the negative part of its domain. More precisely, the number of bids for the price $B^0_t+x$, $x<0$, is given by $f(x)dx$. As before, we can now introduce the {\it extra spread} $D^B:=(D^B_t)_{t\geq0}$. Given a sell order $x_n <0$, a trading time $t_n$, and $D^B_{t_n}$, $D^B_{t_n+}$ is implicitly defined by
\begin{equation}
\int_{D^B_{t_n}}^{D^B_{t_n+}} f(x)dx = x_n.
\end{equation}
Note that $D^B$ is non-positive. We equivalently define the {\it impact on the volume} $E^B:=(E^B_t)_{t\geq0}$ by
\begin{equation}
E^B_{t_n+}:=E^B_{t_n}+x_n.
\end{equation}
$E^B$ is also non-positive, and its connection to $D^B$ is again given by (\ref{eq:relationDE}). To complete the definitions for sell orders, we set $D^B_t:=0$ and $E^B_t:=0$ for all $t\leq t_0$, and
\begin{equation}\label{eq:eq2mod3}
\left\{\begin{array}{ll}
E^B_t := e^{-\rho (t-t_n)}E^B_{t_n+}&\textrm{for Version 1}\\
\ \\
D^B_t := e^{-\rho (t-t_n)}D^B_{t_n+}&\textrm{for Version 2}
\end{array}\right.
\textrm{ for }t\in(t_n,t_{n+1}],
\end{equation}
whereby $t_n$ and $t_{n+1}$ are two successive trading times of the large trader.

Now that all orders are well defined, we introduce the {\it cost} of a large order $x_{t_n}$ at some trading time $t_n$ by
\begin{equation}
\pi_{t_n}(x_{t_n}):=\left\{\begin{array}{ll}
\int_{D^A_{t_n}}^{D^A_{t_n+}}(A^0_{t_n}+x)f(x)dx&\textrm{ for a buy market order $x_{t_n}\geq 0$}\\
\\
\int_{D^B_{t_n}}^{D^B_{t_n+}}(B^0_{t_n}+x)f(x)dx&\textrm{ for a sell market order $x_{t_n}< 0$}
\end{array}\right. .
\end{equation}

We assume that the large trader needs to purchase the $X_0$ shares in $N+1$ steps at equidistant points in time $0=:t_0<\dots <t_N:=T$. His or her {\it admissible strategies} are sequences $\xi=(\xi_0,\dots ,\xi_N)$ of random variables such that
\begin{itemize}
\item $\sum_{n=0}^{N} \xi_n= X_0$,
\item $\xi_n$ is $\mathscr{F}_{t_n}$-measurable for all $n$, and
\item all $\xi_n$ are bounded from below.
\end{itemize}
We denote the set of all admissible strategies by $\hat{\Xi}$. The goal is to find an admissible strategy $\xi^*$ that minimizes the {\it average cost} $\mathscr{C}(\xi)$ given by the sum of the single trades' costs:
\begin{equation}\label{eq:avgCost}
\mathscr{C}(\xi):=\mathbb{E}\left(\sum_{n=0}^N \pi_{t_n}(\xi_n)\right).
\end{equation}
Under the technical assumption that
\begin{equation}\label{eq:infVol}
\lim_{x\to\infty} F(X)=\infty\ \textrm{ and }\ \lim_{x\to -\infty}F(x)=-\infty,
\end{equation}
Alfonsi, Fruth and Schied give the unique optimal strategies for both versions explicitly. We restate them here to give the reader the opportunity to compare them to our theorems for the GAFS model in Section \ref{sec:parameters}. For the sake of convenience, we set $\tau := T/(N+1)=t_{n+1}-t_n$.
\subsubsection*{Optimal strategy for Version 1, Theorem 4.1 in \cite{alfonsi09}}
Suppose that the function
\begin{equation}
h_1(x):=F^{-1}(x)-e^{-\rho\tau}F^{-1}(e^{-\rho\tau}x)
\end{equation}
is one-to-one. Then there exists a unique optimal strategy $\xi^{(1)}=(\xi^{(1)}_0,\dots,\xi^{(1)}_N)$. The initial market order $\xi^{(1)}_0$ is the unique solution of the equation
\begin{equation}\label{eq:toSolve2}
F^{-1}\left(X_0-N\xi^{(1)}_0(1-e^{-\rho\tau})\right)=\frac{h_1(\xi^{(1)}_0)}{1-e^{\rho\tau}},
\end{equation}
the intermediate orders are given by
\begin{equation}
\xi^{(1)}_1=\dots=\xi^{(1)}_{N-1}=\xi^{(1)}_0(1-e^{-\rho\tau}),
\end{equation}
and the final order is determined by
\begin{equation}
\xi^{(1)}_N=X_0-\sum_{n=0}^{N-1} \xi^{(1)}_n.
\end{equation}
In particular, the optimal strategy is deterministic. Moreover, it consists only of nontrivial buy orders, that is $\xi_n>0$ for all $n$.

\subsubsection*{Optimal strategy for Version 2, Theorem 5.1 in \cite{alfonsi09}}
Suppose that the function
\begin{equation}
h_2(x):=x\frac{f(x)-e^{-2\rho\tau}f(e^{-\rho\tau}x)}{f(x)-e^{-\rho\tau}f(e^{-\rho\tau}x)}
\end{equation}
is one-to-one and that the shape function satisfies
\begin{equation}
\lim_{|x|\to\infty}x^2\inf_{y\in [e^{-\rho\tau}x,x]}f(y)=\infty.
\end{equation}
Then there exists a unique optimal strategy $\xi^{(2)}=(\xi^{(2)}_0,\dots,\xi^{(2)}_N)$. The initial market order $\xi^{(2)}_0$ is the unique solution of the equation
\begin{equation}\label{eq:toSolve1}
F^{-1}\left(X_0-N[\xi^{(2)}_0-F(e^{-\rho\tau}F^{-1}(\xi^{(2)}_0))]\right)=h(F^{-1}(\xi^{(2)}_0)),
\end{equation}
the intermediate orders are given by
\begin{equation}
\xi^{(2)}_1=\dots=\xi^{(2)}_{N-1}=\xi^{(2)}_0-F(e^{-\rho\tau}F^{-1}(\xi^{(2)}_0)),
\end{equation}
and the final order is determined by
\begin{equation}
\xi^{(2)}_N=X_0-\sum_{n=0}^{N-1} \xi^{(2)}_n.
\end{equation}
In particular, the optimal strategy is deterministic. Moreover, it consists only of nontrivial buy orders, that is $\xi_n>0$ for all $n$.

\ 

One can easily check that the orders $\xi_1^{(\cdot)},\dots,\xi_{N-1}^{(\cdot)}$ have exactly the volume that has recovered since the last trade. In this sense, the theorems just give the right balance between the first and the last order. This balance is found by solving the particular equations, (\ref{eq:toSolve2}) and (\ref{eq:toSolve1}), given in both theorems.

\section{The Opinion Game}\label{sec:opinionGame}
Next, we focus on the Opinion Game. In Section \ref{sec:model}, we recapitulate the original model as introduced by Bovier, \v{C}ern\'y and Hryniv in \cite{bovier06}. We have already discussed in the introduction why the underlying generalized order book of this model provides even more information about the market behavior than a {\it classical} order book. Yet, the Opinion Game has no explicit notion of orders and, consequently, also large orders and their executions are not defined. However, we argue in Section \ref{sec:reinterpretation} that the generalized order book contains an implicit notion of orders. Furthermore, we state the algorithm that we use to simulate the execution of large orders and show on a qualitative level that this extension leads to a realistic response of the Opinion Game to large orders.

\subsection{The model}\label{sec:model}
We consider a fixed number of traders $N\in\mathbb{N}$ and a fixed number of tradable shares $M<N$. Every trader is described by the pair $(p_i, n_i)$, whereby $p_i$ is the opinion of trader $i$ about the {\it right} logarithmic price; the opinion is individual and subjective. For numerical reasons, $p_i\in\mathbb{Z}$. The number of shares that trader $i$ posseses is given by $n_i$. In the most general setting, $n_i$ can take values form $0$ to $M$; however, we just divide traders in {\it buyers} and {\it sellers} by setting $n_i\in\{0,1\}$. We define the {\it best bid price} by
\begin{equation}
p^b := \max_{i:n_i = 0} p_i,
\end{equation}
and the {\it best ask price} by
\begin{equation}
p^a := \min_{i:n_i = 1} p_i;
\end{equation}
the {\it price} $p$ is given by
\begin{equation}
p := \frac{p^b+p^a}{2}.
\end{equation}
The market is said to be in a {\it stable state} if $p^b < p^a$; no buyer is then willing to pay the lowest asked price and vice versa.

For numerical reasons, the dynamics is defined in discrete time. Every round consists of three steps:
\begin{enumerate}
\item {\bf A trader is chosen}\\
We define
\begin{equation}
g(i,t) := \left\{\begin{array}{ll}
(1+p^b(t)-p_i(t))^{-\gamma}&\textrm{ if trader $i$ is buyer}\\
(1+p_i(t)-p^a(t))^{-\gamma}&\textrm{ if trader $i$ is seller}
\end{array}\right. ,
\end{equation}
and set
\begin{equation}\label{eq:probChosen}
P(\textrm{trader $i$ is chosen at time t}) := \frac{g(i,t)}{\sum_{j=1}^N g(j,t)}.
\end{equation}
The parameter $\gamma > 0$ can be chosen arbitrarily. Observe that the defined measure prefers traders close to the price. The larger $\gamma$ is the greater is this preference. Here, we assume that a trader close to the current price reacts faster to price fluctuations than a long time investor with an opinion being completely different from the current price.
\item {\bf The trader's change of opinion}\\
If a trader is chosen, he or she changes her opinion to $p_i'(t+1):=p_i(t)+d(t)$. The random variable $d(t)$ takes values in $\{-l, -l+1,\dots,l-1, l\}$, $l\in\mathbb{N}$, and is independently sampled for all $t$. The measure of $d(t)$ is given by
\begin{equation}
P(d(t) = m) = \left\{\begin{array}{ll}
\frac{1}{2l+1}\left((\mu_\mathrm{ext}(t)\mu_\cdot)^m \wedge 1\right)&\textrm{for $m\neq 0$}\\
\\
1-\sum_{m=1}^l P(d(t)=\pm m)&\textrm{else}
\end{array}\right.\textrm{, whereby}
\end{equation}
\begin{equation}
\mu_\cdot := \left\{\begin{array}{ll}
\mu_B&\textrm{if trader $i$ is a buyer}\\
\mu_S&\textrm{if trader $i$ is a seller}
\end{array}\right. .
\end{equation}
We assume that $\mu_S < 1 < \mu_B$ to implement the idea that all traders have a tendency to move into the direction of the price. The $\mu_\mathrm{ext}(t)$ introduces a drift that changes randomly in time and acts on all traders in the same way, modeling news, rumors and events influencing the price. This drift process is of paramount importance for the stylized facts, statistical features of the price process on large time scales; however, as we want to concentrate on the large orders' impact, which happens on shorter time scales, we assume $\mu_\mathrm{ext} \equiv 1$ in the remainder of this article.
\item {\bf Trading (if necessary)}\\
If the market with the changed opinion is stable again, that is
\begin{equation}
p^b((p_1(t),\dots,p'_i(t),\dots,p_N(t)))<p^a((p_1(t),\dots,p'_i(t),\dots,p_N(t))),
\end{equation}
we set $p_i(t+1) := p_i'(t+1)$, else a trade happens. Let us assume that trader $i$ is a buyer, the other case is symmetric. We uniformly choose a trading partner $j$ with $p_j(t) = p^a(t)$ and set $n_i(t+1)=0$ and $n_j(t+1)=1$. After the trade, both traders move away from the best price:
\begin{equation}
p_i(t+1) := p^a(t+1)+\bar{g}\ \ \textrm{ and }\ \ p_j(t+1) := p^a(t+1)-g
\end{equation}
whereby $\bar{g}$ and $g$ can be fixed or random numbers in $\mathbb{N}$. This last step is justified by the idea that the traders want to make profit and are only willing to trade for a better price than they have paid.
\end{enumerate}

\subsection{An extension for large orders}\label{sec:reinterpretation}
For the existence of orders in the Opinion Game, let us consider a buyer and a seller with matching opinions such that a trade happens. In order book driven markets,  trades can only come about if both traders have placed some kind of orders. From this point of view, the Opinion Game has an {\it implicit} notion of orders, at least when trades are happening. This observation motivates a change of our point of view on the Opinion Game: In the remainder of this article, we rather think about (maybe hidden or unplaced) buy or sell orders instead of traders with opinion. For the sake of convenience, we omit the word {\it generalized} in the following when we talk about the order book of the Opinion Game.

To test the AFS model, we have to introduce large orders to the Opinion Game. Assume we would like to purchase $X$ stocks at time $t$. Then, we do not apply the standard dynamics explained above at time $t$; instead, we use the following algorithm:
\ \\
\begin{itemize}
\item[] set $p_k^{(1)} := p_k(t)$ for all $k \in \{1,\dots,N\}$
\item[] set $n_k^{(1)} := n_k(t)$ for all $k \in \{1,\dots,N\}$
\item[] let $p^a(1)$ be the best ask price of the configuration $\left(p_k^{(1)},n_k^{(1)}\right)_{k \in \{1,\dots,N\}}$
\item[]
\item[] from $x := 1$ to $X$ do \{
\begin{itemize}
\item[] find $i$ s.th. $p^{(x)}_i \leq p^{(x)}_j$ for all $j \in \{1,\dots,N\}$
\item[] $p^{(x+1)}_i := p^a(x)$
\item[] choose uniformly trading partner $j$ s.th. $i\neq j$ and $p^{(x)}_j = p^a(x)$
\item[] $n^{(x+1)}_i := 1$ and $n^{(x+1)}_j := 0$
\item[] $p^{(x+1)}_j := p^a(x) - g$
\item[] $p^{(x+1)}_i := p^a(x) + \hat{g}(x)$
\item[]
\item[] set $p_k^{(x+1)} := p_k^{(x+1)}$ for all $k \in \{1,\dots,N\}\backslash\{i,j\}$
\item[] set $n_k^{(x+1)} := n_k^{(x+1)}$ for all $k \in \{1,\dots,N\}\backslash\{i,j\}$
\item[] let $p^a(x+1)$ be the best ask price of the configuration $\left(p_i^{(x+1)},n_i^{(x+1)}\right)_{i \in \{1,\dots,N\}}$
\end{itemize}
\item[] \}
\item[] set $p_k(t+1) := p_k^{(X+1)}$ for all $k \in \{1,\dots,N\}\backslash\{i,j\}$
\item[] set $n_k(t+1) := n_k^{(X+1)}$ for all $k \in \{1,\dots,N\}\backslash\{i,j\}$
\end{itemize}

The value $g$ is the same random or deterministic value as in the original dynamics. The random variables $\hat{g}(x)$ are independently distributed with measure
\begin{equation}
P(\hat{g}(x) = k) = \frac{1}{M}\sum_{n=1}^N \mathbbm{1}_{\{p_n^{(x)}-p^a(x)=k\}}\textrm{ for }k\in\mathbb{N}_0.
\end{equation}
In other words, we execute a large buy order of volume $X$ by taking the lowest $X$ orders one by one and putting them directly to the ask price such that a trade is enforced. The number of market participants is constant in the Opinion Game, thus taking orders from the tail is an obvious method to simulate a large order that is placed {\it out of the blue}. After each single trade, we adjust the order prices; the price of the (new) buy order is decreased by $g$, the price of the sell order is increased by $\hat{g}$. The density function of $\hat{g}$ is given by the order book's current shape.

\begin{figure}[ht]
\includegraphics[width=.45\textwidth]{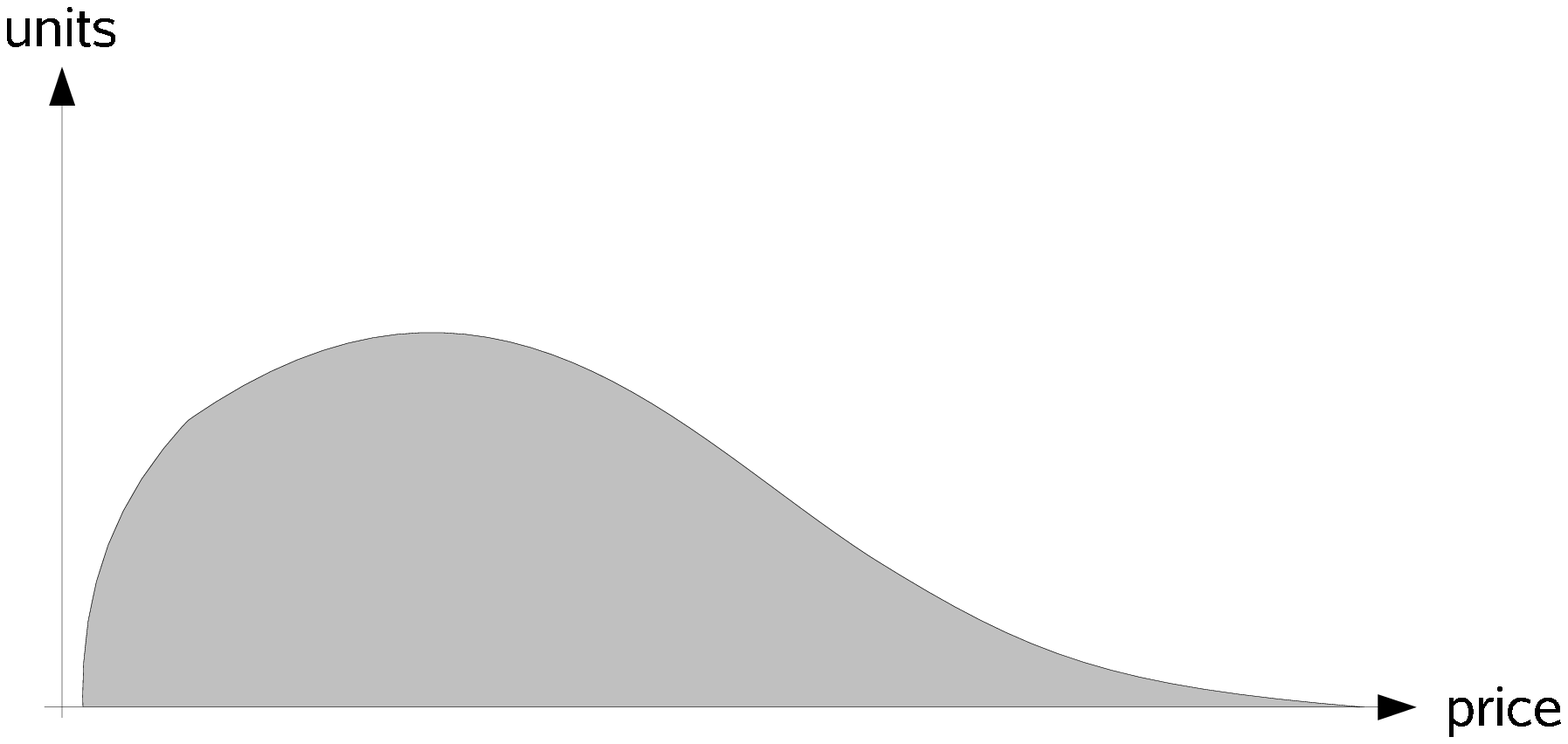}
\includegraphics[width=.45\textwidth]{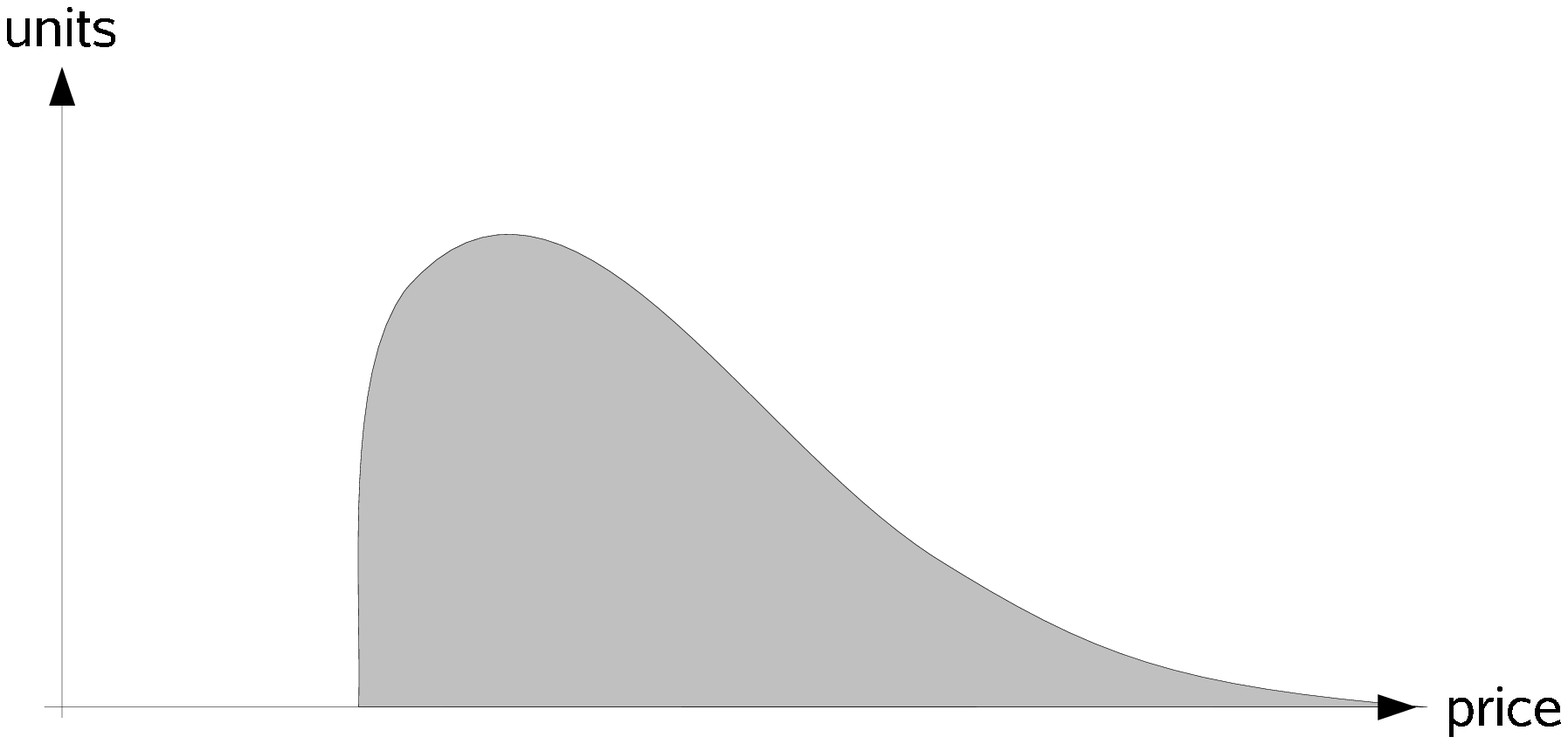}\\
\includegraphics[width=.45\textwidth]{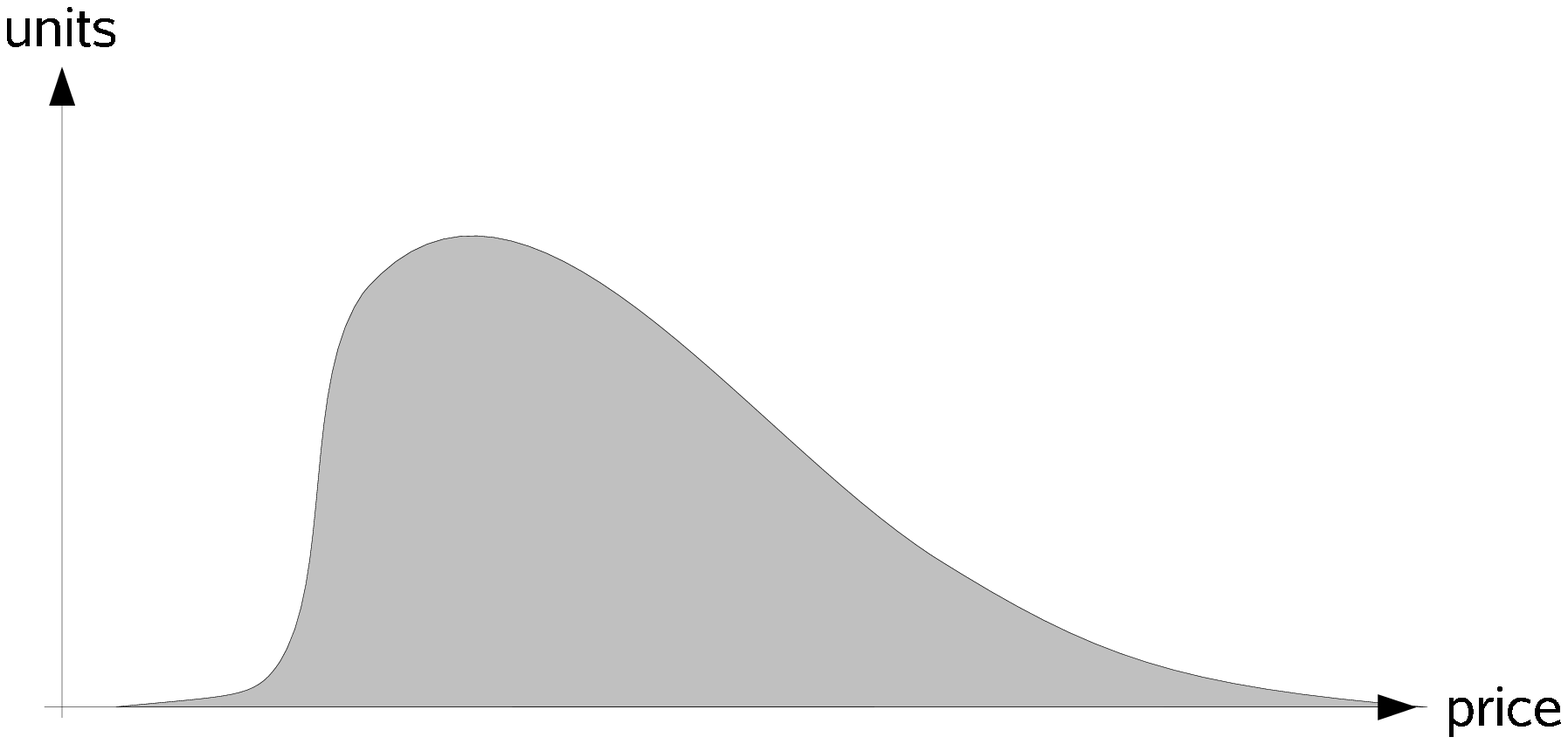}
\caption{Sketch of the order book shape in the Opinion Game when a large order is executed. Before the execution, the order book is in equilibrium (upper left figure); directly afterwards, the best ask price is increased, and there is more liquidity close to it (upper right figure). When the LOB recovers from the order, the best ask price decreases, but the best quotes have a low volume only (lower figure); it takes more time until the order book is in equilibrium again.}
\label{fig:orderBookE}
\end{figure}

This choice of $\hat{g}$ leads to a realistic response of the order book to the execution of large orders (see Figure \ref{fig:orderBookE}). While the large order is executed, the new sell orders have a great probability to be placed in vicinity to the peak of the order book's seller part; thus the peak grows, and the order book provides more liquidity for prices in this region. Here, we implement the idea that the execution of a large buy order leads to a conspicuous rise in the price that attracts more traders to place sell orders close to the current best ask price; these traders hope that the price increase continues such that their orders are executed. At the same time, these additional offers provide more liquidity that slows down the price increase. If we consider the immediate price impact of the large order as function of the executed volume, the additional liquidity leads to a sublinear function shape. Sublinear behavior of an order's immediate price impact has also been observed for real world markets in several empirical studies \cite{bouchaud04,almgren05}. After the execution of the large order, the price increase stops and some traders realize quickly that orders for higher prices will probably not be executed in the near future; they place new orders for lower prices. However, most traders need more time to acknowledge that their price claims are probably too high. In result, the best ask price decreases, but the order book volume in proximity to the new best quote is low. It takes more time until the LOB is back in equilibrium. This recovery behavior of the order book is technically implemented by the preference for traders close to the best quotes in (\ref{eq:probChosen}) when we update opinions. As another feature that is known from real world markets, the best ask price does not return to the value it has had before the execution, but it stabilizes at higher values after the order book has returned to equilibrium. We discuss this {\it permanent impact} on the best price in Section \ref{sec:parameters2}.

Since the dynamics are symmetric, the algorithm applies to large sell orders in the same way.

\section{Determining the parameters}\label{sec:parameters}
The Opinion Game provides a variety of parameters to influence the characteristics of the modeled market. For instance, it is possible to change the size of the market or the volatility in the Opinion Game to simulate {\it different} markets. Nevertheless, we restrict ourselves in the following to one parameter set, which is stated in Subsection \ref{sec:orderBook}. Although the variation of parameters surely leads to additional insight, our choice already gives a sound understanding of the problems that occur when applying the AFS model. In the same subsection, we also describe the {\it averaged} shape of the Opinion Game's order book that will serve as shape function $f$ for the AFS model.

Having set up the Opinion Game, we try to determine the AFS model's resilience speed, $\rho$. It turns out that the assumption of a constant $\rho$ is not valid in the Opinion Game. Thus, we substitute $\rho$ by a function $\bar{\rho}$ that maps both the order's impact and the time elapsed since the last trade to the resilience speed. We describe how we can extract the function values from the sampled data, and argue that it is sufficient to know the impact-dependent  function $\bar{\rho}(\cdot):=\bar{\rho}(\cdot,\tau)$ only; recall that $\tau=T/N$ was the recovery time between two successive trades. Finally, we introduce the generalized AFS theorems that assume the resilience speed to be a function of the price impact (in Version $2$) or the volume impact (in Version $1$).

\subsection{The parameters of the Opinion Game and the shape of its order book}\label{sec:orderBook}
There is a high degree of freedom in the parameters for the Opinion Game. Nevertheless, certain parameter sets have been shown to be more reasonable choices than others. Calibrated with these parameter values, the Opinion Game results in a realistic price process in terms of stylized facts. However, not all choices can be justified rigorously. For an extensive discussion about the choice of parameter $\gamma$, for instance, we refer to \cite{weiss09}. We used the following values in all simulations throughout this article, since those ones have been shown to generate price processes with realistic statistical features \cite{bovier06, weiss09b}:
\begin{center}
\begin{tabular}{|l|p{8.5cm}|}
\hline
$\sharp$ of traders $N$&$2000$\\
\hline
$\sharp$ of shares $M$&$1000$\\
\hline
speed of adaption $\gamma$&$1.5$\\
\hline
jump range $\{-l,\dots,l\}$&$\{-4,\dots,4\}$\\
\hline
drift of buyers $\mu_B$&$e^{0.1}$\\
\hline
drift of sellers $\mu_S$&$e^{-0.1}$\\
\hline
jump ranges $\bar{g}$, $g$&random variables, uniformly distributed on $\{5,\dots,20\}$, sampled idependently every time they are used\\
\hline
\end{tabular}
\end{center}
\ \\
All sample runs that we did in the Opinion Game, either to extract necessary parameters or to test execution strategies, were started independently with a random seed for the random number generator. Furthermore, the recording of data or the execution of large orders was started after $1\,000\,000$ simulation steps only, such that the model had sufficient time to get close to a stable state.

\begin{figure}
\includegraphics[angle=270, width=\textwidth]{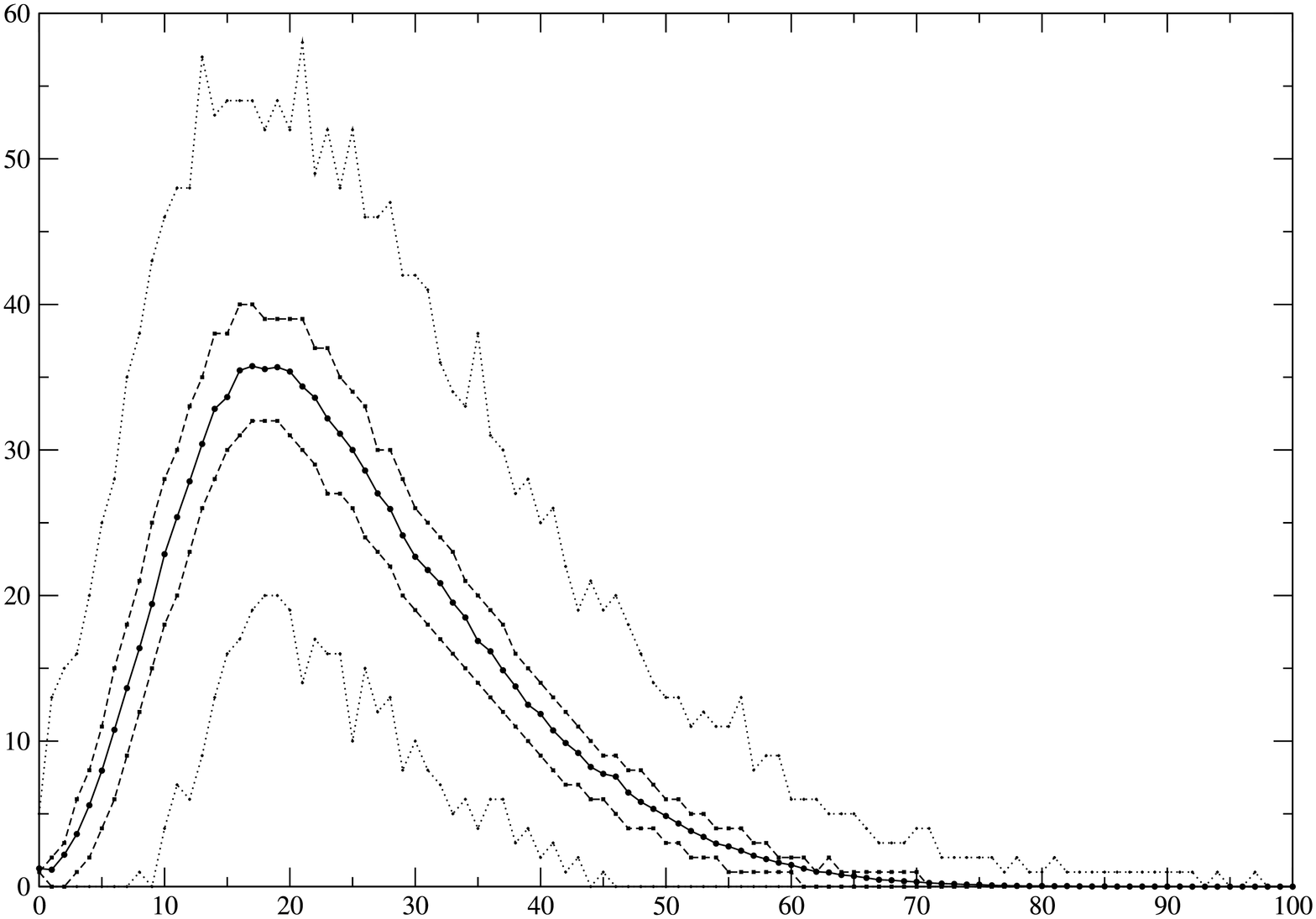}
\caption{The seller part of the LOB relative to the best ask price. The solid line marks the mean values, the dashed lines illustrate the quartiles. The minimal and maximal values are illustrated by the dotted lines.}
\label{fig:orderBook}
\end{figure}
To determine $f$, we recorded $500$ times the Opinion Game's LOB relative to the best prices. Figure \ref{fig:orderBook} shows the resulting upper part of the order book. The lower part is symmetric up to small deviations caused by the object's random nature. Even if the shape is not static as assumed in the AFS model, an {\it averaged shape} is clearly visible. We use these mean values to define the shape function $f$ for the Opinion Game. For non-integer values, we interpolate $f$ by assuming that the function is a right-continuous step function. This means that we violate the assumption of the AFS model about $f$ being continuous. Yet, this choice for $f$ has the advantage that the integral of $f$ from $0$ to an integer $n$ is equal to the sum of the integer function values from $0$ to $n-1$. Furthermore, for all parameter sets that we considered, we were still able to find unique solutions for the optimal trading strategies.

Recall that the price scale in the Opinion Game is logarithmic, whereas the AFS model assumes a linear scale. However, it is possible to scale the grid of the Opinion Game with a factor $\epsilon$, and the difference between logarithmic and linear scale is negligible if $\epsilon$ is small. To determine the order of $\epsilon$, we consider an order of $200$ units of shares, $20\%$ of the market volume in the Opinion Game; it is mentioned in \cite{alfonsi09} that the size of large orders can amount up to twenty percent of the daily traded volume. We assume that the shape of the LOB, $f$, is determined as described above, and the best ask price before our trade is denoted by $A^0$. Then the relative impact costs are given by
\begin{equation}\label{eq:approx}
\frac{1}{200 e^{\epsilon A^0}}\int_0^{D_{0+}}e^{\epsilon(A^0+x)}f(x)dx-1\approx \frac{\epsilon}{200}\underbrace{\int_0^{D_{0+}}x f(x)dx}_{\approx 2039.47}\approx10.20\epsilon.
\end{equation}
An empirical study of the US stock market shows that large orders can cause relative costs up to $3.55\%$ \cite{almgren05}. If we assume that $\epsilon \leq 0.0355/10.2$, then $\epsilon$ is of order $10^{-3}$ at most. Thus, it is reasonable to assume $\epsilon$ to be small. However, we are interested in qualitative results; thus, and for the sake of convenince, we will simply assume that the Opinion Game operates on $\mathbb{Z}$.

\subsection{Determining \texorpdfstring{$\rho$}{rho} for the AFS model}\label{sec:parameters2}
In the following, we present our approach to calibrate $\rho$ for the Opinion Game. We describe our simulation approach and the corresponding results for Version $2$ of the AFS model only. Recall that, in this version, $\rho$ determines the recovery speed of the price impact. Our ansatz and the observations are qualitatively the same for the other version. Nevertheless, we introduce the GAFS model for the price impact dependent as well as for the volume impact dependent case in the end of this section.We first describe how we sampled the necessary data. Afterwards, we focus on the main problems of extracting $\rho$ from those data. Possible solutions are discussed and culminate in this section's main result: The GAFS theorems, which assume that the resilience speed is a function $\bar{\rho}$ depending on the order's impact.

We fixed a price impact $D\in\{1,\dots,20\}$ and ran $2500$ simulations for each value of $D$, resulting in $50\,000$ simulations. Since every run had an initialization period of $1\,000\,000$ steps, each simulation took several seconds. Observe that a simulation time of one second per run already results in a total computing time of almost $14$ hours. As we ran several simulations parallel, we were able to finish the data collecting within a few days.

Each run consisted of a trading part in which a large sell order was executed at once. The particular order volume was determined by its price impact: The trading part was finished as soon as the impact was equal to $D$. In a second experiment's part, we recorded the relaxiation of the price. In particular, the large execution took place at time $\bar{t}:=1\,000\,000$; we recorded
\begin{equation}
\bar{p}(t):=p^a(t+1+\bar{t})-p^a(\bar{t})
\end{equation}
for $t\in\{0,\dots,50\,000\}$. The process $(\bar{p}(t))_{t\in\mathbb{N}_0}$ is the discrete counterpart of the AFS model's process $D^A$.

To avoid problems caused by random fluctuations in $\bar{p}$, we consider the pointwise average of the samples denoted by $\langle\bar{p}\rangle$ and defined by
\begin{equation}
\langle\bar{p}\rangle_t := \frac{1}{2500}\sum_{i=0}^{2500} \bar{p}^i_t
\end{equation}
for all $t\in\{0,\dots,50\,000\}$, $\bar{p}^i$ denoting the $i$th sample. For a clear distinction, we denote the value for $\rho$ that we extract from $\langle\bar{p}\rangle$ by $\bar{\rho}_\mathrm{num}$. The AFS model assumes $\langle\bar{p}\rangle$ to be of the form
\begin{equation}\label{eq:naivRho}
\langle\bar{p}\rangle_t = De^{-\bar{\rho}_\mathrm{num}t}
\end{equation}
with a static value $\bar{\rho}_\mathrm{num}$; this follows from equation (\ref{eq:eq2mod2}). Thus we should be able to determine $\bar{\rho}_\mathrm{num}$ by
\begin{equation}\label{eq:detRho}
\bar{\rho}_\mathrm{num} = \frac{\ln D - \ln \langle\bar{p}\rangle_t}{t}
\end{equation}
for an arbitrary $t$. However, the right hand side of the equation depends on $D$ and $t$; thus, we would like to consider $\bar{\rho}_\mathrm{num}(D,t)$ as a function.

\begin{figure}
\includegraphics[angle=270, width=.6\textwidth]{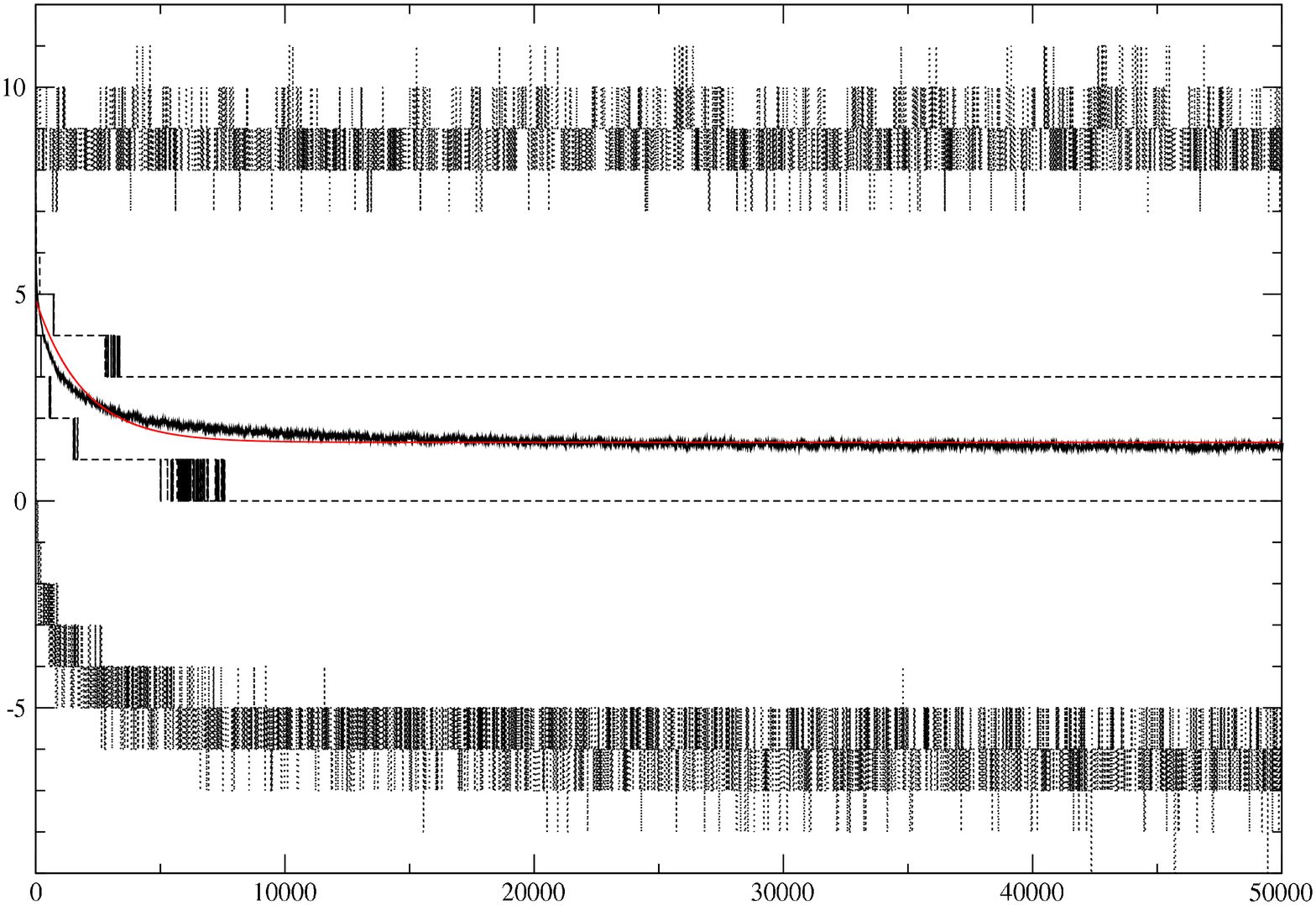}\\
\includegraphics[angle=270, width=.6\textwidth]{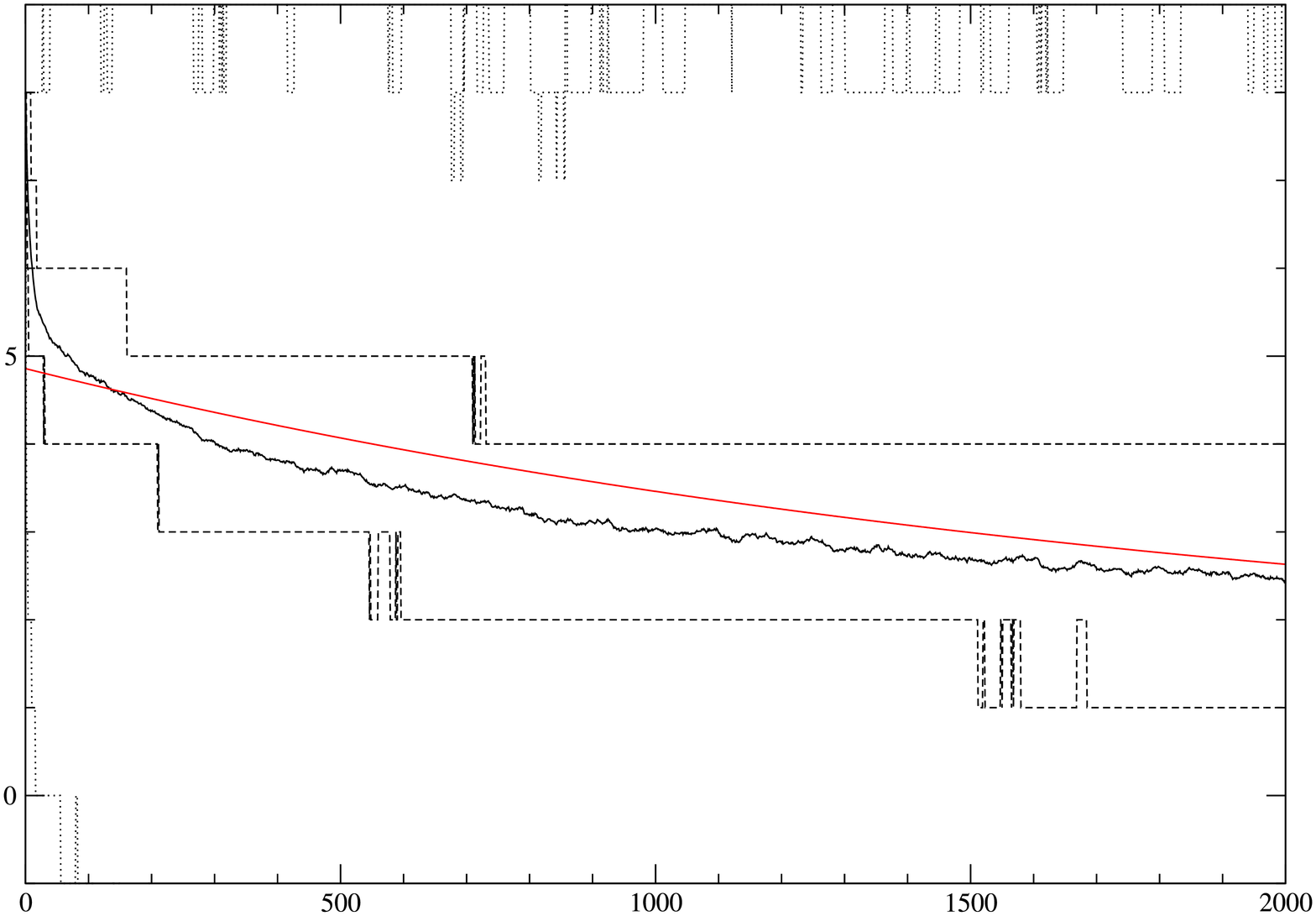}\\
\includegraphics[angle=270, width=.6\textwidth]{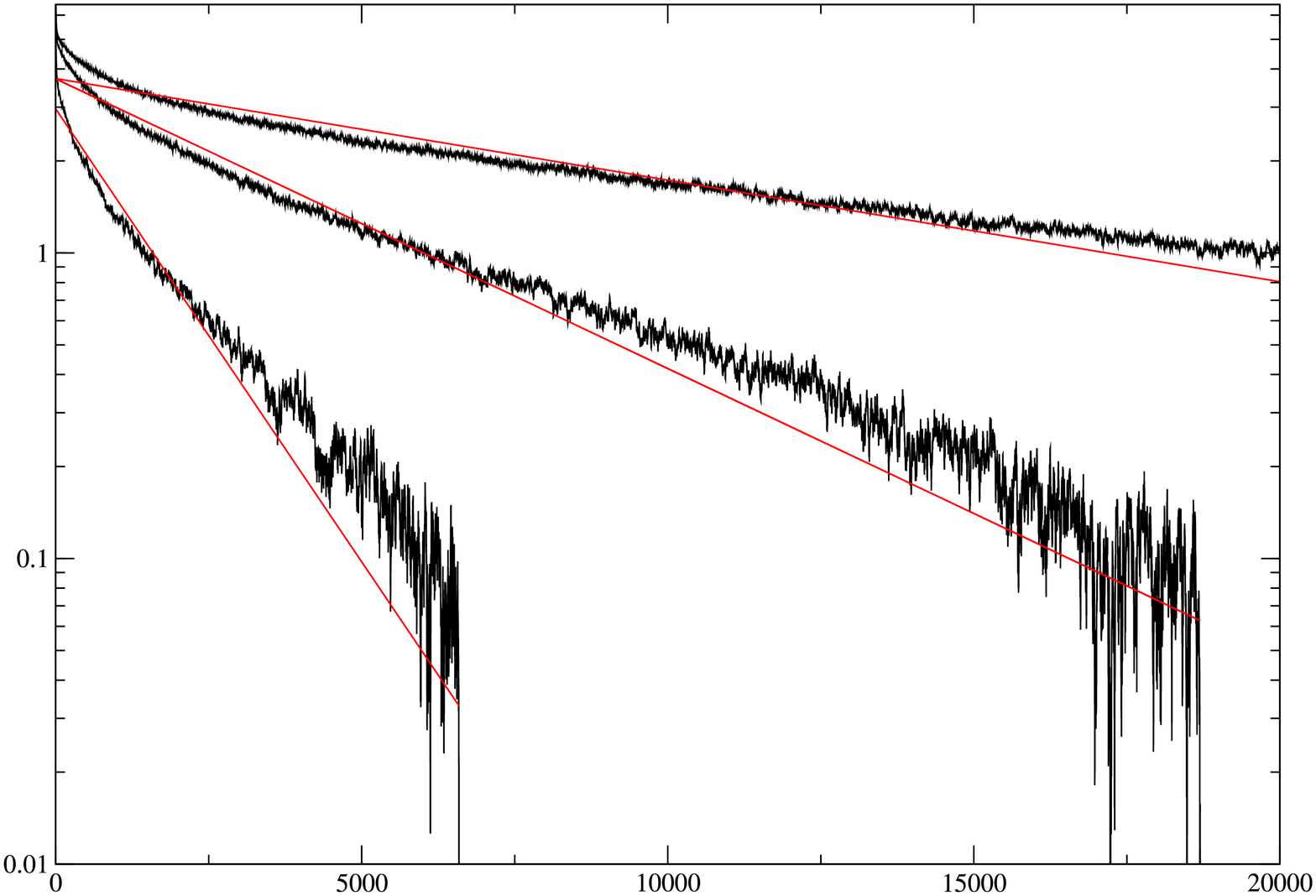}
\caption{The two upper graphs show quartiles and extremal values of $2500$ samples of $\bar{p}$ for $D=8$, and the corresponding $\langle\bar{p}\rangle$ and $\hat{p}$ (red). The graph on top illustrates the long time behavior on the domain $t\in [0,50\,000]$. Clearly, $\hat{p}$ converges to a level $A_\mathrm{D}>0$. The middle graph displays $t \in [0,2000]$ showing the poor approximation by $\hat{p}$. The lower graph shows $\langle\bar{p}\rangle$ (black) for $D=16$, $D=12$ and $D=8$ (top down) as well as their regression functions $\hat{p}$ (red) on a logarithmic scale and with respect to the new asymptotic levels $A_\mathrm{D}$. If $\bar{\rho}_\mathrm{num}$ was constant the $\langle\bar{p}\rangle$ should be approximately parallel.}
\label{fig:rho1}
\end{figure}

Given $\langle\bar{p}\rangle$, let $\hat{p}:[0,\infty)\to\mathbb{R}$ the corresponding regression function of the form 
\begin{equation}\label{eq:canRho}
\hat{p}_t := A+Be^{\hat{\rho} t},
\end{equation}
for $t\in [0,\infty)$. It is determined by a Newton-Gau\ss{} algorithm with three degrees of freedom: $A$, $B$, $\hat{\rho}$. Observe that all three values can depend on $D$. The form of the regression function is motivated by assumption (\ref{eq:naivRho}), which also leads to the expectation that $A=0$ and $B=D$. Figure \ref{fig:rho1} shows the statistical behavior of $\bar{p}$ for $D=8$, the corresponding $\langle\bar{p}\rangle$ and $\hat{p}$. Furthermore, we compare $\langle\bar{p}\rangle$ for different $D$ values. The three main problems are visible:
\begin{enumerate}
\item The AFS model assumes $A_\mathrm{D}$ to be $0$; this is not the case.
\item The measured data is only well-approximated by an exponential function for large times. For small $t$, it is doubtful that the assumption of an exponential decay is the right choice at all.
\item If $\bar{\rho}_\mathrm{num}$ was constant the $\langle\bar{p}\rangle$ should be approximately parallel on a logarithmic scale; instead, $\bar{\rho}_\mathrm{num}$ depends $D$.
\end{enumerate}
These problems occured for all tested values of $D$. Next, we discuss the problems and their consequences for the determination of $\bar{\rho}_\mathrm{num}$ one by one.

\subsubsection{Existence of a permanent price impact}
\begin{figure}
\includegraphics[angle=270, width=.85\textwidth]{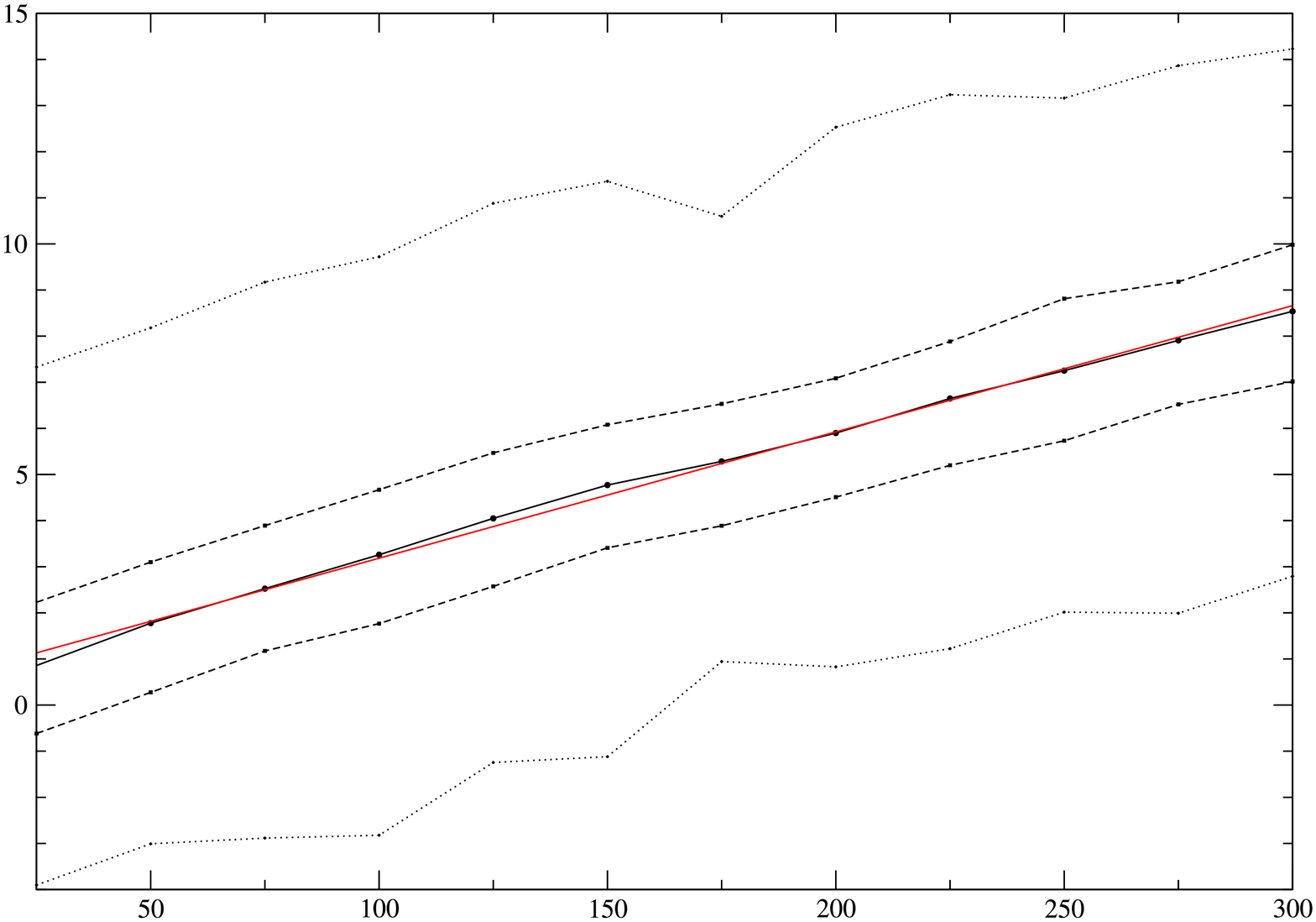}
\caption{Mean, quartiles and extremal values of $500$ samples of the permanent impact in dependence on the purchased volume $V\in\{25,50,\dots,275,300\}$. For every volume $V$, we recorded the best ask price before the trade and the averaged best ask price $500\,000$ steps after the trade. Here, the averaged best ask price is the mean of the best ask price sampled all $100$ steps over a time interval of $100\,000$ steps. The linear regression of the mean is displayed in red.}
\label{fig:perImp}
\end{figure}

The reason for problem (1) is a permament impact on the order book that a large trade causes. After having recovered, the LOB is shifted by $I_\mathrm{per}(X)$, whereby $I_\mathrm{per}:\mathbb{R}\to\mathbb{R}$ is assumed to be increasing and $I_\mathrm{per}(0)=0$. Huberman and Stanzl \cite{huberman04} argued on a theoretic level that linearity of $I_\mathrm{per}$ is equivalent to the absence of arbitrage opportunities. Empirical studies by Almgren et al. \cite{almgren05} reinforce the conjecture of a linear permanent impact: The authors state that the permanent impact is well described by the power law $x^{0.9\pm 0.1}$ with respect to a Gaussian error model; the assumption of linearity cannot be rejected by this result. Figure \ref{fig:perImp} shows the permanent impact for the Opinion Game. The mean is well aproximated by a linear function with coefficient $0.02738$.

Concerning the problems in determining $\bar{\rho}_\mathrm{num}$, caused by the positive $A_\mathrm{D}$, we have two possibilities: First, we could ignore the permanent impact such that $\bar{\rho}_\mathrm{num}$ would be given by (\ref{eq:detRho}). This would be an appropriate solution for small $t$, but it would cause the AFS model to assume that even for large $t$ the LOB is still not close to equilibrium; $\bar{\rho}_\mathrm{num}$ could become arbitrarily small. Second, we could assume that the whole model has been shifted by $A_\mathrm{D}$ such that $A_\mathrm{D}$ is the new zero line. In this case, $\bar{\rho}_\mathrm{num}$ would be given by
\begin{equation}
\bar{\rho}_\mathrm{num}(D,t) = \frac{\ln D-\ln(\langle\bar{p}\rangle_t-A_\mathrm{D})}{t},
\end{equation}
which is fine for large $t$ but grows to infinity as $t$ goes to zero. To avoid this problem, we define
\begin{equation}\label{eq:rho}
\bar{\rho}_\mathrm{num}(D,t) := \frac{\ln D-\ln(\langle\bar{p}\rangle_t-(1-e^{-t})A_\mathrm{D})}{t}.
\end{equation}
Furthermore, let us point out that there is no special reason to choose $1-\exp(-t)$. However, at this point, it becomes clear that the complex dynamics within the LOB are poorly described by an added permanent impact function.

\subsubsection{$\langle\bar{p}\rangle$ is poorly approximated by an exponential function}
Since $\langle\bar{p}\rangle$ should decay exponentially fast, $\bar{\rho}_\mathrm{num}$ should be a constant. However, the existence of a permanent impact and the consequential definition of $\bar{\rho}_\mathrm{num}$ in (\ref{eq:rho}) makes the validity of this assumption unlikely here. Yet, even without the permanent impact, the description of $\langle\bar{p}\rangle$ by an exponential function is poor as the upper right graph of Figure \ref{fig:rho1} shows. As mentioned in Section \ref{sec:AFSmodel}, the rejection of an exponential decay does not contradict former research results. If we nevertheless try to calibrate $\rho$, it becomes time-dependent. A time-dependent resilience speed seems to be incompatible with the AFS theorems at first, but a closer look at the theorem's statement reveals that $\rho$ is only needed to determine the order book state {\it before} the next trade, given the state {\it after} the current trade. The time between two succeeding trades is given by $\tau$. Thus, we focus on $\bar{\rho}_\mathrm{num}(\cdot,\tau)$ and use the notation
\begin{equation}
\bar{\rho}_\mathrm{num}(D):=\bar{\rho}_\mathrm{num}(D,\tau)
\end{equation}
assuming that $\tau$, which is given by the input parameters $N$ and $T$, is fixed. Figure \ref{fig:rho2} shows the function $\bar{\rho}_\mathrm{num}(\cdot,\tau)$ for several values of $\tau$.
\begin{figure}
\includegraphics[angle=270, width=.65\textwidth]{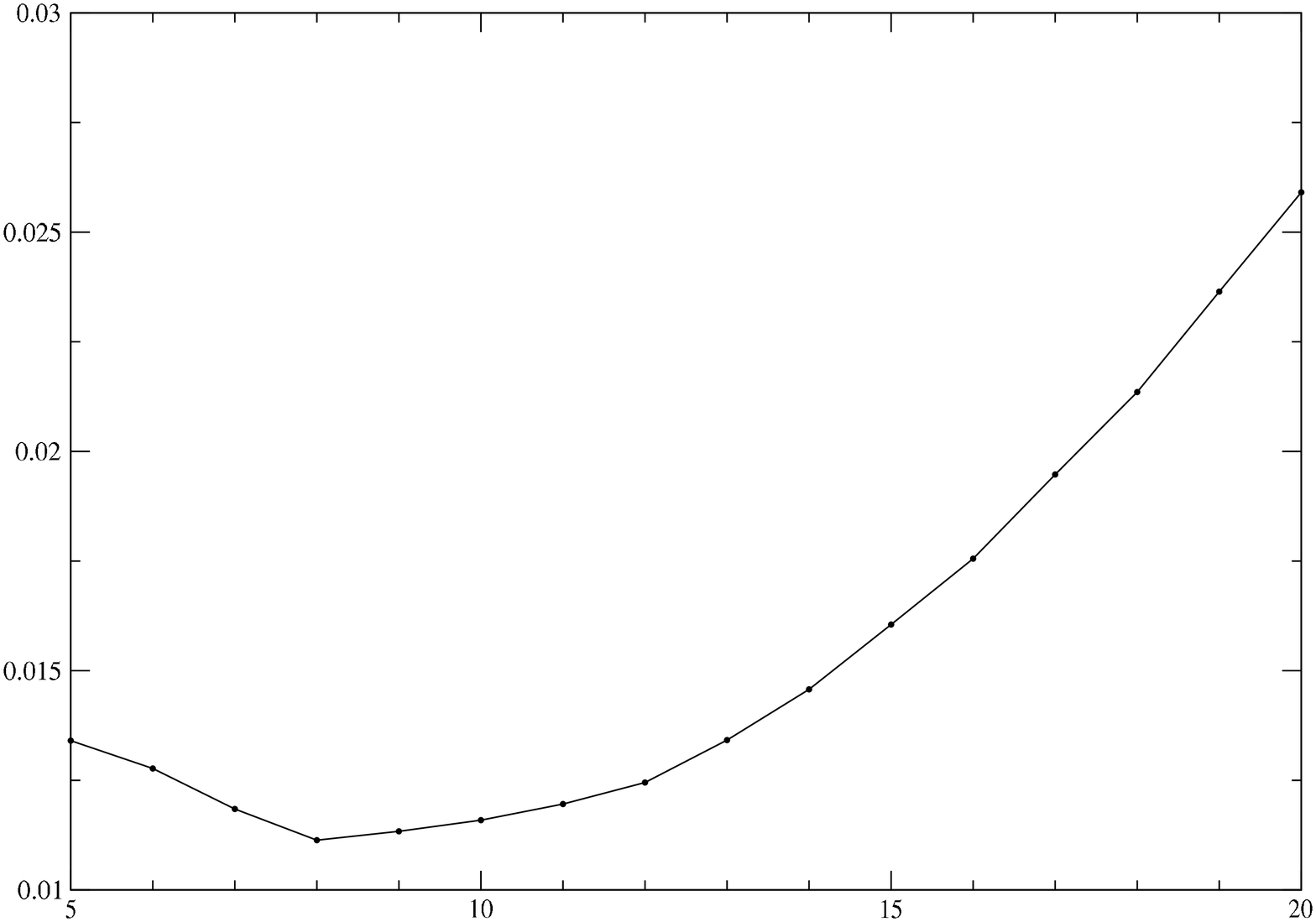}\\
\includegraphics[angle=270, width=.65\textwidth]{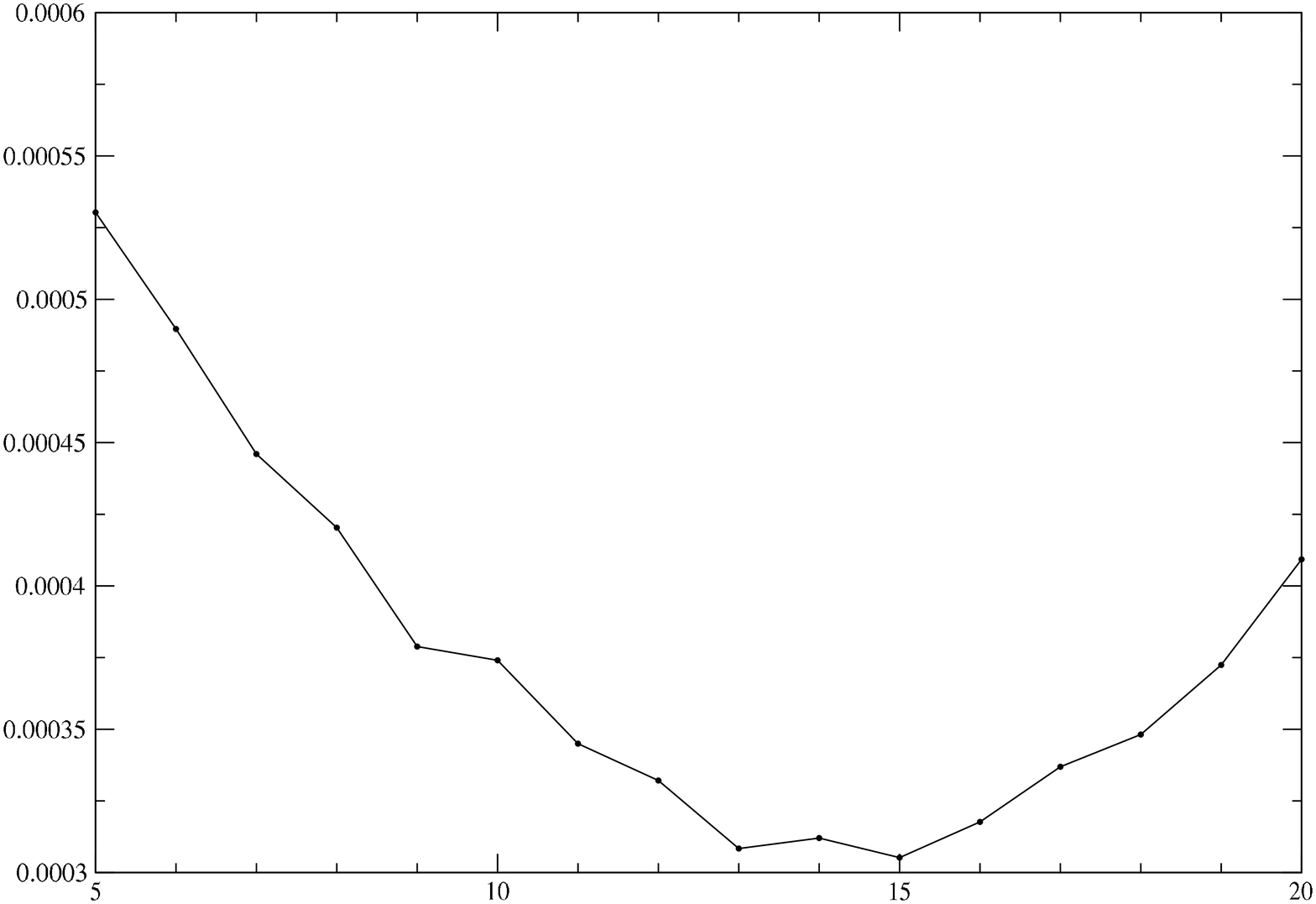}\\
\includegraphics[angle=270, width=.65\textwidth]{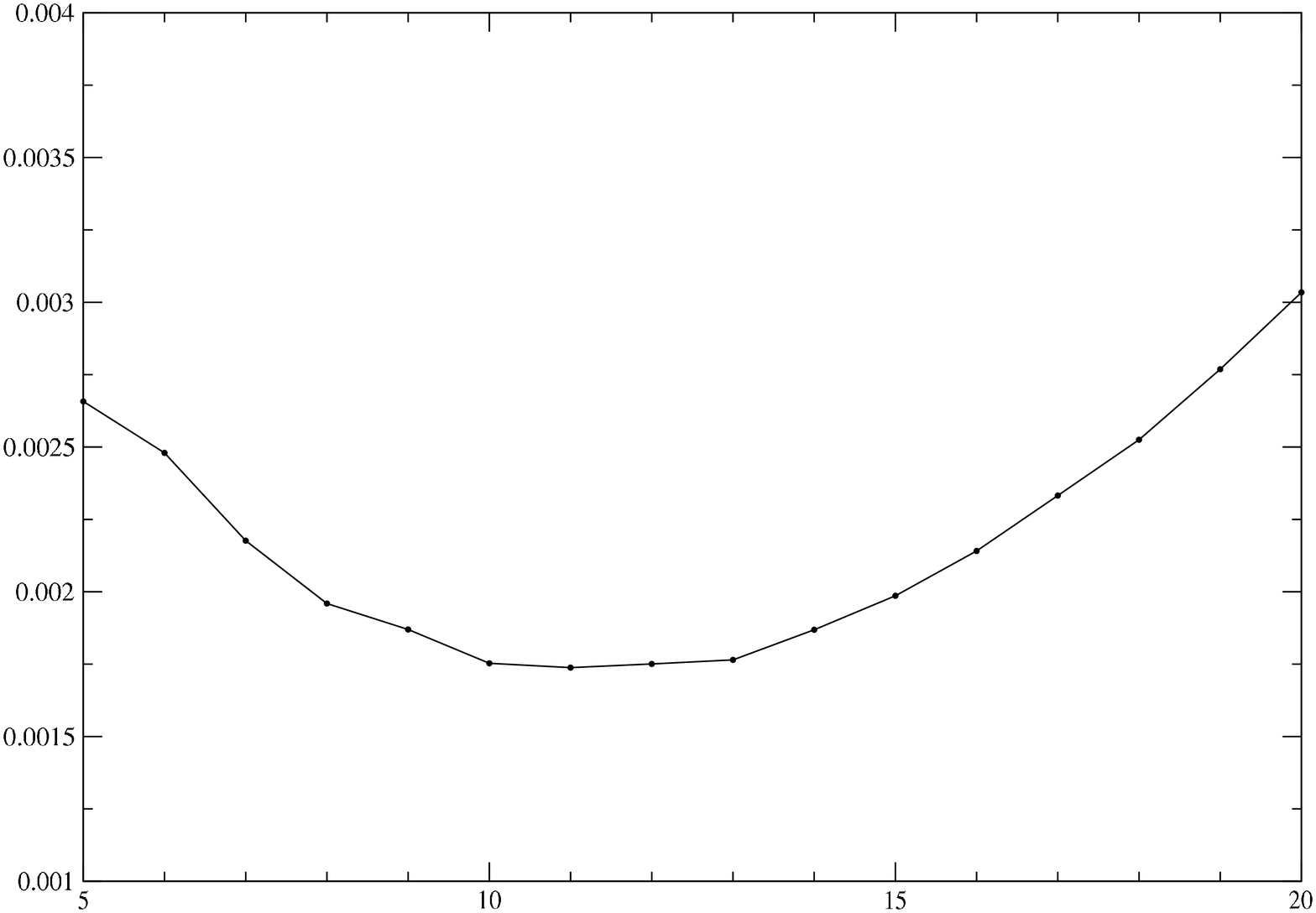}
\caption{We show the graphs of $\rho_\mathrm{num}(D,\tau)$ in dependence on $D$ for $\tau=70,700,7000$ (top down). Observe that the x-axis only begin in $5$ due to the fact that small price impacts cannot be distinguished from the noise contained in the signal.}
\label{fig:rho2}
\end{figure}

\subsubsection{$\bar{\rho}_\mathrm{num}$ is a function of $D$}
In contrast to the time dependence, the dependence on the order's price impact requires a generalization of the AFS theorem, stated in Section \ref{sec:AFSmodel}. Now, the resilience speed $\bar{\rho}:\mathbb{R}\to (0,\infty)$ is a continously differentiable function of $D^A$. In particular, the equations (\ref{eq:eq2mod2}) and (\ref{eq:eq2mod3}), which describe the price recovery in the AFS model, change to
\begin{equation}
D^A_t := e^{-\bar{\rho}(D_{t_n+}^A)(t-t_n)D_{t_n+}^A}\textrm{ for }t\in(t_n, t_{n+1}],
\end{equation}
\begin{equation}
D^B_t := e^{-\bar{\rho}(D_{t_n+}^B)(t-t_n)D_{t_n+}^B}\textrm{ for }t\in(t_n, t_{n+1}].
\end{equation}
We denote this modified model as {\it Version 2} of the {\it generalized AFS model}.

For the following theorem concerning the optimal trading strategy for the GAFS model, we need two technical assumptions:
\begin{equation}\label{eq:rhoAss1}
\textrm{The range of }\bar{\rho}\textrm{ is a subset of }[k,K],\ 0<k<K<\infty,\textrm{ and}
\end{equation}
\begin{equation}\label{eq:rhoAss2}
1-\tau\bar{\rho}'(x)x > 0\textrm{ for all }x\in\mathbb{R}.
\end{equation}
The first assumption bounds the relaxation speed, the second assumption ensures that a larger impact cannot overtake a smaller one in the recovery phase as we will see in Lemma \ref{lem:inBet}.
\begin{thm}[Optimal stratey for the generalized AFS model, Version 2]\label{thm:weiss2}
Suppose that $\bar{\rho}$ fulfils (\ref{eq:rhoAss1}) and (\ref{eq:rhoAss2}), and that $f$ satisfies
\begin{equation}\label{eq:fCond2}
\lim_{|x|\to\infty}x^2\inf_{y\in[e^{\tau\bar{\rho}(x)}x,x]}f(y)=\infty.
\end{equation}
Furthermore, let the function
\begin{equation}
h_2(x):=x\frac{f(x)-e^{-2\tau\bar{\rho}(x)}f(e^{-\tau\bar{\rho}(x)}x)(1-\tau\bar{\rho}'(x)x)}{f(x)-e^{-\tau\bar{\rho}(x)}f(e^{-\tau\bar{\rho}(x)}x)(1-\tau\bar{\rho}'(x)x)}
\end{equation}
be one-to-one. Then there exists a unique optimal strategy $\xi^{(2)}=(\xi^{(2)}_0,\dots,\xi^{(2)}_N)\in\hat{\Xi}$. The initial market order $\xi^{(2)}_0$ is the unique solution of the equation
\begin{equation}\label{eq:mainEq2}
F^{-1}\left(X_0-N\left[\xi^{(2)}_0-F\left(e^{-\tau\bar{\rho}(F^{-1}(\xi^{(2)}_0))}F^{-1}(\xi^{(2)}_0)\right)\right]\right)=h_2(F^{-1}(\xi^{(2)}_0)),
\end{equation}
the intermediate orders are given by
\begin{equation}\label{eq:inOrders2}
\xi^{(2)}_1=\dots=\xi^{(2)}_{N-1}=\xi^{(2)}_0-F\left(e^{-\tau\bar{\rho}(F^{-1}(\xi^{(2)}_0))}F^{-1}(\xi^{(2)}_0)\right),
\end{equation}
and the final order is determined by
\begin{equation}
\xi^{(2)}_N=X_0-\sum_{n=0}^N \xi^{(2)}_n.
\end{equation}
In particular, the optimal stratey is deterministic. Moreover, it consists only of nontrivial buy orders, that is $\xi^{(2)}_n>0$ for all $n$.
\end{thm}
\begin{proof}
See Appendix \ref{sec:proof2}.
\end{proof}
Observe that the intermediate orders of the optimal strategy, defined in (\ref{eq:inOrders2}), have the same size. Furthermore, they suggest to purchase exactly that volume that has recovered since the last trade. The GAFS model has inherited this feature from the AFS model. Yet, this observation means that also the $D^A_{t_n+}$ are equal to each other for all $n\in\{0,\dots,N-1\}$, and thus, $\bar{\rho}$ is only evaluated for one value. In other words, although $\bar{\rho}$ is a function, the optimal strategy {\it uses} only one value. Of course, if $\bar{\rho}\equiv \rho$ for some constant $\rho$ in the GAFS model both models, the GAFS and the AFS, coincide. This is the main advantage of the GAFS theorem: It determines the {\it right} resilience speed from $\bar{\rho}$; a manual calibration, as in the AFS model, is not needed anymore.

As already mentioned, our simulations for the calibration of $\rho$ in Version $1$ of the AFS model result in the same problems as described for Version $2$. Especially, $\rho$ becomes volume impact dependent, motivating the GAFS model, Version $1$: Now, the resilience speed $\bar{\rho}:[0,\infty)\to (0,\infty)$ is a {\it twice} differentiable function of $E^A$. In particular, the equations (\ref{eq:eq2mod1}) and (\ref{eq:eq2mod3}) from the AFS model become
\begin{equation}
E^A_t := e^{-\bar{\rho}(E^A_{t_n+})(t-t_n)}E^A_{t_n+},\ t\in(t_n,t_{n+1}],
\end{equation}
\begin{equation}
E^B_t := e^{-\bar{\rho}(E^B_{t_n+})(t-t_n)}E^B_{t_n+},\ t\in(t_n,t_{n+1}].
\end{equation}
in the GAFS model. Then, the following theorem determines the optimal trading strategy in the set of all admissible strategies $\hat{\Xi}$:

\begin{thm}[Optimal stratey for the generalized AFS model, Version 1]\label{thm:weiss1}
Suppose that $\bar{\rho}$ fulfils the assumptions (\ref{eq:rhoAss1}) and (\ref{eq:rhoAss2}), and additionally
\begin{equation}\label{eq:rhoCond1}
e^{-\bar{\rho}(x)\tau}\left(1-\tau\bar{\rho}'(x)x\right)<1\ \textrm{ for all $x\in\mathbb{R}$}.
\end{equation}
Furthermore, let the function
\begin{equation}
h_1(x):=\frac{F^{-1}(x)-e^{-\bar{\rho}(x)\tau}\left(1-\tau\bar{\rho}'(x)x\right)F^{-1}(e^{-\bar{\rho}(x)\tau}x)}{1-e^{-\bar{\rho}(x)\tau}\left(1-\tau\bar{\rho}'(x)x\right)}
\end{equation}
be one-to-one. Then there exists a unique optimal strategy $\xi^{(1)}=(\xi^{(1)}_0,\dots,\xi^{(1)}_N)\in\hat{\Xi}$. The initial market order $\xi^{(1)}_0$ is the unique solution of the equation
\begin{equation}
F^{-1}\left(X_0-N\xi^{(1)}_0(1-e^{-\bar{\rho}(\xi^{(1)}_0)\tau})\right)=h_1(\xi^{(1)}_0),
\end{equation}
the intermediate orders are given by
\begin{equation}
\xi^{(1)}_1=\dots=\xi^{(1)}_{N-1}=\xi^{(1)}_0(1-e^{-\bar{\rho}(\xi^{(1)}_0)\tau}),
\end{equation}
and the final order is determined by
\begin{equation}
\xi^{(1)}_N=X_0-\sum_{n=0}^N \xi^{(1)}_n.
\end{equation}
In particular, the optimal stratey is deterministic. Moreover, it consists only of nontrivial buy orders, that is $\xi_n>0$ for all $n$.
\end{thm}
\begin{proof}
See Appendix \ref{sec:proof1}.
\end{proof}
As in Version $2$ of the GAFS model, $\bar{\rho}$ is only evaluated in one value, and if $\bar{\rho}\equiv\rho$ the best strategies of the GAFS and the AFS models coincide.

\section{Numerical results}\label{sec:numerics}
Let us turn to the numerical results of this paper. Again, we focus on Version $2$ and use the parameter values determined in the last section to calculate the GAFS optimal strategies and to apply them in the Opinion Game. We show first that the resulting costs show an {\it expected} behavior on a general level, and that the AFS model with a suboptimal value for $\rho$ suggests a strategy that produces significantly higher costs than the corresponding GAFS strategy. Afterwards, we compare the costs sampled in the Opinion Game to the costs predicted by the GAFS model, and find large differences. We refer to the values for $f$ and $\bar{\rho}$, $\bar{\rho}_\mathrm{num}$, as determined in the Sections \ref{sec:orderBook} and \ref{sec:parameters2}.
\begin{table}
\begin{tabular}{|r|r||r|r|r||r|r|r|}
\hline
$N$ & $T$ & $\xi^{(2)}_0$ & $\xi^{(2)}_1$ & $\xi^{(2)}_N$ & Predicted & Sampled & Samp/Pred\\
\hline
\hline
$40$ & $400$ & $8.95$ & $4.74$ & $6.38$ & $701.47$ & $1867.74$ & $266\%$\\
\hline
$40$ & $4000$ & $6.13$ & $4.81$ & $6.15$ & $500.24$ & $1573.50$ & $315\%$\\
\hline
$40$ & $40\,000$ & $5.16$ & $4.86$ & $5.40$ & $392.42$ & $1076.89$ & $274\%$\\
\hline
\hline
$50$ & $400$ & $8.29$ & $3.80$ & $5.29$ & $691.94$ & $1853.37$ & $268\%$\\
\hline
$50$ & $4000$ & $5.20$ & $3.88$ & $4.94$ & $462.51$ & $1535.96$ & $332\%$\\
\hline
$50$ & $40\,000$ & $4.26$ & $3.90$ & $4.82$ & $349.26$ & $1014.42$ & $290\%$\\
\hline
\hline
$80$ & $400$ & $7.55$ & $2.37$ & $5.63$ & $691.65$ & $1832.69$ & $265\%$\\
\hline
$80$ & $4000$ & $3.73$ & $2.44$ & $3.37$ & $387.98$ & $1464.03$ & $377\%$\\
\hline
$80$ & $40\,000$ & $2.65$ & $2.46$ & $2.91$ & $231.67$ & $914.17$ & $395\%$\\
\hline
\end{tabular}
\caption{The optimal strategies according to the GAFS model, Version $2$, for $X=200$ and several values for $N$ and $T$.}
\label{tab:table2a}
\end{table}

Table \ref{tab:table2a} shows the GAFS optimal strategies and their costs for different values of $T$ and $N$. We consider two kinds of costs. The {\it predicted costs} are the impact costs that are theoretically predicted by the (G)AFS model. Here, we assume that the market behaves as described in Section \ref{sec:AFSmodel}. The {\it sampled costs} are the average of $500$ samples with the given strategy in the Opinion Game. Observe first that the predicted and the sampled costs decrease if the trading time or the number of trading opportunities increase. Of course, this is no special feature of the (G)AFS strategies; every fixed strategy benefits from a larger $\tau$, which is implied by a greater $T$, and additional trading opportunities can be used, but do not have to be used. Thus, every reasonable strategy can only perform better with larger $T$ or $N$. Nevertheless, the costs of the GAFS strategies show a {\it reasonable} behavior.

Furthermore, the GAFS strategies perform better than the AFS strategies: Recall that the AFS model with the right value for $\rho$ results in the same optimal strategy as the GAFS model. Moreover, the (G)AFS model assumes an exponential decay of the price impact (see (\ref{eq:naivRho})). We have taken this assumption into account by introducing $\langle \bar{\rho}\rangle$'s regression function $\hat{p}$ in (\ref{eq:canRho}), which was of the form
\begin{equation}
\hat{p}_t := A+Be^{-\hat{\rho}t}.
\end{equation}
Table \ref{tab:table2b} shows the optimal strategies and their costs for $(X,T,N)=(200,4000,80)$ with respect to the AFS model with $\rho=\hat{\rho}$ and the GAFS model with $\bar{\rho}=\bar{\rho}_\mathrm{num}$. The example shows that a naive guess of a good $\rho$ can lead to much higher costs: The AFS costs amount $253\%$ of the GAFS costs in prediction, and still $108\%$ in the samples.
 
The last two paragraphs have shown that the GAFS strategies are reasonable and superior to the AFS stratgies. However, returning to Table \ref{tab:table2a}, we see that the predicted and the sampled costs for the individual parameter sets differ strongly. The last column shows both kinds of costs in relation to each other. Obviously, the sampled costs are multiple times higher. This observation is a strong evidence that the assumptions of the (G)AFS model are insufficient to capture the whole complexity of the order book dynamics in the Opinion Game. It is doubtful if the (G)AFS model really suggests optimal trading strategies for this artificial market environment. With regard to the Opinion Game features concerning the order book behavior that we have discussed in Section \ref{sec:opinionGame}, it is highly unlikely that the (G)AFS strategies minimize the costs in real world markets.
\begin{table}
\begin{tabular}{|r||r|r|r||r|r|}
\hline
Strategy & $\xi^{(2)}_0$ & $\xi^{(2)}_1$ & $\xi^{(2)}_N$ & Predicted & Sampled\\
\hline
\hline
GAFS & $3.73$ & $2.44$ & $3.37$ & $387.98$ & $1464.03$\\
\hline
AFS & $21.02$ & $0.97$ & $102.26$ & $979.97$ & $1584.12$\\
\hline
\end{tabular}
\caption{The optimal strategies and their costs for the AFS model with $\rho=\hat{\rho}$ and the GAFS model with $\bar{\rho}$ from Section \ref{sec:parameters2}. $(X,T,N)=(200,4000,80)$.}\label{tab:table2b}
\end{table}

\section{Conclusions}
In this paper, we have tried to apply the AFS model to an artificial market environment. The elegance of the AFS model, the order book approach and the explicit results for the optimal trading strategies, cannot be denied. Yet, the problems we faced in calibrating the model to our market pose the question if the AFS model assumptions are oversimplified. We point out again that the problems we had to handle, the permanent impact and the non-exponential decay of the impact, are {\it not} artificial. It is well known that those effects are also characteristic for real markets. Even if it is possible to bypass some problems or to extend the model in a suitable way as we did by introducing the GAFS model, the question remains if a next generation of large order models is necessary. In \cite{gatheral10}, the authors call the large order market models with an underlying order book models of the second generation dissociation of (first generation) large order models working with fixed price impact functions as described in the introduction. Here, a third generation of models is conceivable taking into account that the order book shape is not constant such that there is no one-to-one correspondence between the price and the volume and that the market adapts to periodically executed large orders. Yet, it is also obvious those models would be of much higher complexity and analytic results would be hard to get; a wide subject for future research.

\textbf{Acknowledgement.} We thank Anton Bovier (University of Bonn), who suggested the article's subject to us and supported our research with many remarks, and Antje Fruth (QPL Berlin) for the discussions concerning the AFS model. Especially, we thank Evangelia Petrou (University of Bonn), who contributed to the article's quality with a countless number of critical comments, suggestions and queries.

\appendix
\section{Proofs of the Theorems \ref{thm:weiss2} and \ref{thm:weiss1}}\label{sec:proof}
The structure of the proofs remains the same as in the proofs of the corresponding AFS theorems (see Appendices A to C in \cite{alfonsi09}). Nevertheless, we need to justify the constraints on $\bar{\rho}$; furthermore, the calculations become more complicated by our generalization. For simplicity, we assume $t_0 = 0$ in this section.

We start with the introduction of slightly changed dynamics for the GAFS model and the reduction of the admissible strategies to deterministic ones. For any admissible strategy $\xi$, the new dynamics is defined by the processes $D:=(D_t)_{t\geq 0}$ and $E:=(E_t)_{t\geq 0}$. We set $D_0 = D_{t_0} := 0 =: E_{t_0} = E_0$ and
\begin{equation}
E_{t_n+}:=E_{t_n}+\xi_n\ \textrm{and}\ D_{t_n+}:=F^{-1}(F(E_{t_n})+\xi_n)
\end{equation}
for the trading times $t_0,\dots,t_N$. The processes' values between two successive trading times $t\in (t_n,t_{n+1})$ are given by
\begin{equation}
\begin{array}{ll}
E_t := e^{-\bar{\rho}(E_{t_n+})(t-t_n)}E_{t_n+}&\textrm{for Version } 1;\\
D_t := e^{-\bar{\rho}(D_{t_n+})(t-t_n)}D_{t_n+}&\textrm{for Version } 2.
\end{array}
\end{equation}
Given one process, we can recover the other one by the equations (\ref{eq:relationDE}):
\begin{equation}
\begin{array}{ccc}
E_t = F(D_t)&\textrm{and}&D_t=F^{-1}(E_t).
\end{array}
\end{equation}

\begin{lemma}\label{lem:inBet}
Under assumption (\ref{eq:rhoAss2}),
\begin{equation}\label{eq:inBetween}
E^B_t\leq E_t\leq E^A_t\textrm{ and }D^B_t\leq D_t\leq D^A_t
\end{equation}
for all $t\geq 0$. In the special case that all $\xi_n$ are non-negative, we have $D^A=D$ and $E^A=E$.
\end{lemma}
\begin{proof}
To see that $D^A=D$ and $E^A=E$ if $\xi$ consists of buy orders only, observe that the new dynamics matches exactly the original ones for such a $\xi$.

For the general case, we consider $E^B_t\leq E_t$; the other inequalities follow equivalently. Observe that it is sufficient to prove
\begin{equation}
E^B_{t_{n+1}}\leq E_{t_{n+1}}
\end{equation}
for
\begin{equation}
E^B_{t_n+}\leq E_{t_n+},
\end{equation}
since both function are exponentially decreasing on $(t_n,t_{n+1}]$, and the relative order of $E^B$ and $E$ cannot be reversed from $t_n$ to $t_n+$. Furthermore, we can restrict ourselves to the case that $E^B_{t_n+}$ and $E_{t_n+}$ have the same sign, as the signs canno change in the considered time interval, $(t_n, t_{n+1}]$. We are consequently done if we can show that the inequality
\begin{equation}\label{eq:staysBigger}
|x+y|e^{-\bar{\rho}(x+y)\tau} > |x|e^{-\bar{\rho}(x)\tau}
\end{equation}
for all $(x,y)\in\mathbb{R}^2$ with $\mathrm{sgn}(x+y)$=$\mathrm{sgn}(x)$ and $|x+y| > |x|$. Observe that we have equality in the equation above if we consider the trivial case that $y=0$. We define a function $u_x:\mathbb{R}\to\mathbb{R}$ by
\begin{equation}
u_x(y) := (x+y)e^{-\bar{\rho}(x+y)\tau}.
\end{equation}
Differentiation yields
\begin{equation}
u'_x(y) = e^{-\bar{\rho}(x+y)\tau}(1-\tau\bar{\rho}'(x)x).
\end{equation}
The right hand side of this equation is positive by assumption (\ref{eq:rhoAss2}), thus $u_x$ is strictly increasing. Since $u_x(0) = xe^{-\bar{\rho}(x)\tau}$, (\ref{eq:staysBigger}) is proven.
\end{proof}

It remains to define the {\it simplified price of $\xi_n$ under the new dynamics} by
\begin{equation}
\bar{\pi}_{t_n}(\xi_n):=\int_{D_{t_n}}^{D_{t_n+}} (A^0_{t_n}+x)f(x)dx = A^0_{t_n}\xi_n+\int_{D_{t_n}}^{D_{t_n+}} xf(x)dx.
\end{equation}
Observe that
\begin{equation}\label{eq:barpi}
\bar{\pi}_{t_n}(\xi_n) \leq \pi_{t_n}(\xi_n)
\end{equation}
for all admissible strategies $\xi$ because of Lemma \ref{lem:inBet}. In particular, if $\xi$ consists of buy orders only, we have equality.

We show in the next two sections that, the strategies given in the Theorems \ref{thm:weiss1} and \ref{thm:weiss2}, $\xi^{(1)}$ and $\xi^{(2)}$, are the unique minimizers of the {\it price functional}
\begin{equation}
\bar{\mathscr{C}}(\xi):=\mathbb{E}\left[\sum_{n=0}^N \bar{\pi}_{t_n}(\xi_n)\right]
\end{equation}
for the corresponding version of the model. As $\xi^{(1)}$ and $\xi^{(2)}$ consist of buy orders only, (\ref{eq:barpi}) and the remark afterwards imply that these strategies are also the minimizers of the original price functional $\mathscr{C}$.

We turn to the reduction of $\hat{\Xi}$ to deterministic strategies. Let us define the {\it remaining trading volume} $X=(X_t)_{t\in [0,T]}$ by
\begin{equation}\label{eq:defTradingVol}
X_t :=\left\{\begin{array}{ll}
X_0-\sum_{t_n<t}\xi_n&\textrm{for }t\leq T\\
0&\textrm{for }t>T  
\end{array}\right. .
\end{equation}
Furthermore, we set $X_{t_{N+1}}:=0$. We can transform the price of a strategy $\xi\in\hat{\Xi}$ by
\begin{equation}\label{eq:priceTransform}
\sum_{n=0}^N \bar{\pi}_{t_n}(\xi_n) = \sum_{n=0}^N A^0_{t_n}\xi_n+\sum_{n=0}^N \int_{D_{t_n}}^{D_{t_n+}}xf(x)dx,
\end{equation}
and use definition (\ref{eq:defTradingVol}) as well as {\it integration by parts} to rewrite the first term on the right hand side:
\begin{equation}\label{eq:priceTransform2}
\sum_{n=0}^N A^0_{t_n}\xi_n = -\sum_{n=0}^N A^0_{t_n}\left(X_{t_{n+1}}-X_{t_n}\right)=X_0A_0+\sum_{n=0}^N X_{t_n}\left(A^0_{t_n}-A^0_{t_{n+1}}\right).
\end{equation}
Since $\xi$ is admissible, $X$ is a bounded process and $X_{t_n}$ is $\mathscr{F}_{t_{n-1}}$-measurable, $A^0$ is a martingale. Thus, the expectation of (\ref{eq:priceTransform2}) must be $X_0A_0$. The second term on the right hand side of (\ref{eq:priceTransform}) is deterministic for a given realization of a strategy $\xi(\omega)$. We denote this term by
\begin{equation}
\begin{array}{llll}
C^{(i)}(\xi):&\mathbb{R}^{N+1}&\to&\mathbb{R}\\
&\xi&\mapsto&\sum_{n=0}^N \int_{D_{t_n}}^{D_{t_n+}}xf(x)dx
\end{array}
\end{equation}
for Version $i$, $i\in\{1,2\}$. Now, we can express $\bar{\mathscr{C}}$ by
\begin{equation}
\bar{\mathscr{C}}(\xi) = A_0X_0+\mathbb{E}(C^{(i)}(\xi)).
\end{equation}
We spend the next two section to show $C^{(i)}$ has a unique minimizer in the set
\begin{equation}
\Xi := \left\{x:=(x_0,\dots,x_N)\in\mathbb{R}^{N+1}:\sum_{n=0}^N x_n = X_0\right\}
\end{equation}
and this minimizer is determined by the formula given in Theorem \ref{thm:weiss1} or \ref{thm:weiss2} respectively.\\
For the sake of convenience, we introduce some more notation:
\begin{equation}
\bar{a}_x := \exp(-\tau\bar{\rho}(x))\textrm{ for }x\in\mathbb{R},
\end{equation}
\begin{equation}
a_n:=\left\{\begin{array}{ll}
\exp(-\tau\bar{\rho}(E_{t_n+}))&\textrm{in Section \ref{sec:proof1}}\\
\exp(-\tau\bar{\rho}(D_{t_n+}))&\textrm{in Section \ref{sec:proof2}}
\end{array}\right.
\textrm{ for }n\in\{0,\dots,N\}.
\end{equation}
Because the range of $\bar{\rho}$ is assumed to be $[k,K]$, $0<k<K<\infty$, by (\ref{eq:rhoAss1}),
\begin{equation}
\begin{array}{ccc}
e^{-\tau K}\leq \bar{a}_x\leq e^{-\tau k}&\textrm{and}&e^{-\tau K}\leq a_n\leq e^{-\tau k}.
\end{array}
\end{equation}
Additionally, we will need these functions:
\begin{equation}
\begin{array}{ccc}
\tilde{F}(x) := \int_0^x xf(x)dx&\textrm{and}&G(x):=\tilde{F}(F^{-1}(x)).
\end{array}
\end{equation}
Observe that
\begin{equation}
G'(x)=\tilde{F}'(F^{-1}(x))(F^{-1})'(x)=F^{-1}f(F^{-1}(x))\frac{1}{f(F^{-1}(x))}=F^{-1}(x),
\end{equation}
and thus, G is twice continuously differentiable, non-negative, convex and has a fixed point in $0$.
\subsection{The optimal strategy for Version 1}\label{sec:proof1}
In this section, we calculate the unique minimizer of $C^{(1)}$ in $\Xi$. For any $\xi=(x_0,\dots,x_N)\in\Xi$, we have
\begin{eqnarray}\label{eq:costTransform}
C^{(1)}(\xi)&=&\sum_{n=0}^N\int_{D_{t_n}}^{D_{t_n+}}xf(x)dx\\
&=&\sum_{n=0}^N\left[\tilde{F}(F^{-1}(E_{t_n+}))-\tilde{F}(F^{-1}(E_{t_n}))\right]\nonumber\\
&=&\sum_{n=0}^N\left[G(E_{t_n}+x_n)-G(E_{t_n})\right]\nonumber
\end{eqnarray}
\begin{lemma}\label{lem:minExists1}
The function $C^{(1)}$ has at least one local minimum in $\Xi$.
\end{lemma}
\begin{proof}
The statement will follow from
\begin{equation}
C^{(1)}(\xi)\to\infty\textrm{ for }||\xi||_\infty\to\infty
\end{equation}
because $C^{(1)}$ is continuous. First, we use the properties of $G$ to find a lower bound for $G(x)-G(cx)$, $x\in\mathbb{R}$ and $c\in [0,1]$:
\begin{eqnarray}\label{eq:gDiff}
G(x)-G(cx)&\geq&G(cx)+(x-cx)G'(cx)-G(cx)\\ 
&=&(1-c)|F^{-1}(cx)||x|.\nonumber
\end{eqnarray}
The inequality (\ref{eq:gDiff}) applied to (\ref{eq:costTransform}) leads to a lower bound for $C^{(1)}$:
\begin{eqnarray}
C^{(1)}(\xi)&=&G(E_{t_N}+x_N)-G(E_{t_0})+\sum_{n=0}^{N-1}\left[G(E_{t_n}+x_n)-G(E_{t_{n+1}})\right]\nonumber\\
&=&G\left(\left(\Pi_{n=0}^N a_n\right)x_0+\dots+a_{N-1}x_{N-1}+x_N\right)-G(0)\nonumber\\
&&+\sum_{n=0}^{N-1}\Big[G\left(\left(\Pi_{m=0}^{n-1} a_m\right)x_0+\dots+a_{n_1}x_{n-1}+x_n\right)\\
&&\phantom{+\sum_{n=0}^{N-1}\Big[}-G\left(a_n\left[\left(\Pi_{m=0}^{n-1} a_m\right)x_0+\dots+a_{n_1}x_{n-1}+x_n\right]\right)\Big]\nonumber\\
&\geq&G\left(\left(\Pi_{n=0}^N a_n\right)x_0+\dots+a_{N-1}x_{N-1}+x_N\right)-G(0)\nonumber\\
&&+\sum_{n=0}^{N-1}\Big[(1-a_n)\left|F^{-1}\left(a_n\left[\left(\Pi_{m=0}^{n-1} a_m\right)x_0+\dots+a_{n_1}x_{n-1}+x_n\right]\right)\right|\nonumber\\
&&\phantom{+\sum_{n=0}^{N-1}\Big[(1-a_n)}\left|a_n\left[\left(\Pi_{m=0}^{n-1} a_m\right)x_0+\dots+a_{n_1}x_{n-1}+x_n\right]\right|\Big].\nonumber
\end{eqnarray}
We define a linear mapping $T:\mathbb{R}^{N+1}\to\mathbb{R}^{N+1}$ by
\begin{equation}
T(\xi) := \left(x_0, a_0x_0+x_1,\dots,\left[\Pi_{n=0}^{N-1}a_n\right]x_0+\dots+a_{N-1}x_{N-1}+x_N\right),
\end{equation}
and the smallest $a_n$ by
\begin{equation}
a:=\min\{a_n:n\in\{0,\dots,N\}\}.
\end{equation}
Observe that
\begin{equation}
\left|\left|T(\xi)\right|\right|_\infty\geq \left|\left|\left(x_0, ax_0+x_1,\dots,a^nx_0+\dots+ax_{N-1}+x_N\right)\right|\right|_\infty\to\infty
\end{equation}
for $||\xi||_\infty\to\infty$ as well as $G(x)\to\infty$ and $|F^{-1}(a x)||x|\to\infty$ for $|x|\to\infty$. The last statement follows, because $F$ is unbounded. Finally, we define
\begin{equation}
H(x) := \min\left(G(x),|F^{-1}(a x)||x|\right).
\end{equation}
Also $H(x)\to\infty$ for $x\to\infty$, and consequently,
\begin{equation}
C^{(1)}(\xi) \geq H(||T(\xi)||_\infty)-G(0)\to\infty.
\end{equation}
\end{proof}

One has to determine $\xi_0^{(1)}$ by solving
\begin{equation}
F^{-1}\left(X_0-N\xi^{(1)}_0(1-a_0)\right)=h_1(\xi^{(1)}_0)
\end{equation}
in Theorem \ref{thm:weiss1}. We define the function
\begin{equation}
\hat{h}_1(x):=h_1(x)-F^{-1}\left(X_0-N(1-\bar{a}_x)x\right)
\end{equation}
for which $\xi_0^{(1)}$ is a zero.
\begin{lemma}\label{lem:zero}
Given that the assumptions of Theorem \ref{thm:weiss1} hold, function $\hat{h}_1$ has at most one zero, which is positive if it exists.
\end{lemma}
\begin{proof}
For the existence of at most one zero, it is sufficient to show that $\hat{h}_1$ is strictly increasing. The function $h_1$ has a fixed point in $0$, is positive for positive arguments and continuous as well as bijective, thus it must be strictly increasing or, equivalently, its slope must be strictly positive. Consequently, the slope of $\hat{h}_1$ is also positive, because
\begin{eqnarray}
\hat{h}_1'(x)&=&h_1'(x)-\frac{d}{dx}\left[F^{-1}\left(X_0-Nx(1-\bar{a}_x)\right)\right]\\
&=&h_1'(x)+N\frac{1-\bar{a}_x\left(1-\tau\bar{\rho} '(x)x\right)}{f(F^{-1}(X_0-Nx(1-\bar{a}_x)))},
\end{eqnarray}
and the numerator of the second term is positive by assumption (\ref{eq:rhoCond1}). The positivity of the zero (if existing) follows simply from
\begin{equation}
\hat{h}_1(0) = -F^{-1}(X_0)<0.
\end{equation}
\end{proof}

Now, we are prepared to prove Theorem \ref{thm:weiss1}.
\begin{lemma}
Strategy $\xi^{(1)}$ is the unique minimizer of function $C^{(1)}$ and all components of $\xi^{(1)}$ are positive.
\end{lemma}
\begin{proof}
We showed in Lemma \ref{lem:minExists1} that there is an optimal strategy $\xi^*=(x_0^*,\dots,x_N^*)\in\Xi$. Thus, there must be a Lagrange multiplier $\nu\in\mathbb{R}$ such that
\begin{equation}
\frac{\partial}{\partial x^*_n}C^{(1)}(\xi^*)=\nu\textrm{\ \ \ for }n\in\{0,\dots,N\}.
\end{equation}
Using representation (\ref{eq:costTransform}) of $C^{(1)}$, one gets
\begin{equation}
\frac{\partial}{\partial x_n} C^{(1)}(x) = F^{-1}(E_{t_n+})+a_n\left(1-\bar{\rho} '(E_{t_n+})E_{t_n+}\right)\left[\frac{\partial}{\partial x_{n+1}} C^{(1)}(x)-F^{-1}(E_{t_{n+1}})\right]
\end{equation}
for $n\in\{0,\dots,N-1\}$. In combination with the Langrange multiplier, the recursive formula yields
\begin{equation}
F^{-1}(E_{t_n+})+a_n\left(1-\bar{\rho} '(E_{t_n+})E_{t_n+}\right)\left[\nu-F^{-1}(E_{t_{n+1}})\right]=\nu
\end{equation}
\begin{equation}
\Leftrightarrow \nu = \frac{F^{-1}(E_{t_n+})-a_n\left(1-\bar{\rho} '(E_{t_n+})E_{t_n+}\right)F^{-1}(a_n E_{t_n+})}{1-a_n\left(1-\bar{\rho} '(E_{t_n+})E_{t_n+}\right)}=h_1(E_{t_n+})
\end{equation}
for $n\in\{0,\dots,N-1\}$. The function $h_1$ is bijective by assumption, and thus,
\begin{eqnarray}
x^*_0&=&h_1^{-1}(\nu)\\
x^*_n&=&(1-a_0)x^*_0\textrm{ for }n\in\{1,\dots,N-1\}\label{eq:xnFromx0}\\
x^*_N&=&X_0-x^*_0-(N-1)x^*_0(1-a_0).
\end{eqnarray}
Therefore, the optimal strategy $\xi^*$ is completely defined if we can determine $x^*_0$. By (\ref{eq:costTransform}),
\begin{eqnarray}
C^{(1)}(x^*)&=&G(x^*_0)-G(0)+(N-1)\left[G(a_0x^*_0+(1-a_0)x^*_0)-G(a_0x^*_0)\right]\nonumber\\
&&+G(a_0x^*_0+X_0-x^*_0-(N-1)(1-a_0)x^*_0)-G(a_0x^*_0)\nonumber\\
&=&N\left[G(x_0^*)-G(a_0x^*_0)\right]+G(X_0-N(1-a_0)x_0^*)-G(0)\\
&=:&C_0^{(1)}(x^*_0).\nonumber
\end{eqnarray}
We know that $C^{(1)}_0$ has a minimum because of Lemma \ref{lem:minExists1}. We can find it by differentiation:
\begin{eqnarray}
&&\frac{d}{dx}C_0^{(1)}(x)\label{eq:diffC0}\\
&=&N\big[F^{-1}(x)-\bar{a}_x[1-\tau\bar{\rho}'(x)x]F^{-1}(\bar{a}_xx)\nonumber\\
&&\phantom{N\big[}-\left(1-\bar{a}_x[1-\tau\bar{\rho}'(x)x]\right)F^{-1}(X_0-N(1-\bar{a}_x)x)\big]\nonumber\\
&=&N\left(1-\bar{a}_x[1-\tau\bar{\rho}'(x)x]\right)\hat{h}_1(x).\nonumber
\end{eqnarray}
Assumption (\ref{eq:rhoCond1}) and Lemma \ref{lem:zero} tell us $C^{(1)}$ has exactly one minimum, and this minimum is positive. We have established the uniqueness and representation of the optimal strategy.

It remains to show that all components of $x^*$ are positive. We already know that $x^*_0 >0$. The positivity of $x_n^*$ follows from (\ref{eq:xnFromx0}) for all $n\in\{1,\dots,N-1\}$. For the last order, $x_N^*$, observe that (\ref{eq:diffC0}) vanishes in $x^*_0$. Furthermore, $F^{-1}$ is strictly increasing, and thus,
\begin{eqnarray}
0&=&F^{-1}(x^*_0)-a_0(1-\tau\bar{\rho}'(x^*_0)x^*_0)F^{-1}(a_0x^*_0)\nonumber\\
&&-\left[1-a_0(1-\tau\bar{\rho}'(x^*_0)x^*_0)\right]F^{-1}(\underbrace{X_0-N(1-a_0)x^*_0}_{=x^*_N+a_0x^*_0})\\
&>&\left[1-a_0(1-\tau\bar{\rho}'(x^*_0)x^*_0)\right]\left[F^{-1}(a_0x^*_0)-F^{-1}(a_0x^*_0+x^*_N)\right],
\end{eqnarray}
which, indeed, implies the positivity of $x^*_N$.
\end{proof}

\subsection{The optimal strategy for Version 2}\label{sec:proof2}
In this section, we determine the unique minimizer of $C^{(2)}$ in $\Xi$. For $\xi=(x_0,\dots,x_N)\in\Xi$, we have
\begin{eqnarray}\label{eq:costTransform2}
C^{(2)}(\xi)&=&\sum_{n=0}^N\int_{D_{t_n}}^{D_{t_n+}}xf(x)dx\\
&=&\sum_{n=0}^N\left(G(x_n+F(D_{t_n}))-\tilde{F}(D_{t_n})\right)\nonumber
\end{eqnarray}
\begin{lemma}\label{lem:minExists2}
The function $C^{(2)}$ has a local minimum in $\Xi$.
\end{lemma}
\begin{proof}
Again, it suffices to show
\begin{equation}
C^{(2)}(\xi) \to\infty\textrm{ for }||\xi||_\infty\to\infty.
\end{equation}
We rearrange (\ref{eq:costTransform2}) and get
\begin{eqnarray}\label{eq:costTransform2b}
&&C^{(2)}(\xi)\\
&=&\sum_{n=0}^N\left(\tilde{F}(F^{-1}(x_n+F(D_{t_n})))-\tilde{F}(D_{t_n})\right)\nonumber\\
&=&\tilde{F}(a_NF^{-1}(x_N+F(D_{t_N})))\nonumber\\
&&+\sum_{n=0}^N\left(\tilde{F}(F^{-1}(x_n+F(D_{t_n})))-\tilde{F}(a_nF^{-1}(x_n+F(D_{t_n})))\right)\nonumber\\
&\geq&\sum_{n=0}^N\left(\tilde{F}(F^{-1}(x_n+F(D_{t_n})))-\tilde{F}(a_nF^{-1}(x_n+F(D_{t_n})))\right)\nonumber
\end{eqnarray}
A lower bound for the last line of (\ref{eq:costTransform2b}) is given by
\begin{eqnarray}
\tilde{F}(x)-\tilde{F}(\bar{a}_xx)&=&\left|\int_{\bar{a}_xx}^x zf(z)dz\right|\\
&\geq& \inf_{y\in [\bar{a}_xx, x]} f(y)\left|\int_{\bar{a}_xx}^xz dz\right|\nonumber\\
&=&\frac{1}{2}(1-\bar{a}_x^2)x^2\inf_{y\in [\bar{a}_xx, x]} f(y)\ \geq\ 0.\nonumber
\end{eqnarray}
Because of the assumptions (\ref{eq:fCond2}) and (\ref{eq:infVol}), we know
\begin{equation}
H(x) := \frac{1}{2}(1-\bar{a}_{F^{-1}(x)}^2)(F^{-1}(x))^2\inf_{y\in [\bar{a}_{F^{-1}(x)}F^{-1}(x), F^{-1}(x)]} f(y).
\end{equation}
tends to infinity for $|x|\to\infty$. Finally, we introduce the mapping
\begin{equation}
T(x) := (x_0, x_1+F(D_{t_1}),\dots,x_N+F(D_{t_N})),
\end{equation}
for which $C^{(2)}(x)\geq H(||T(x)||_\infty)$ holds. It remains to show that $||T(x)||_\infty \to \infty$ for $|x|\to\infty$. Let us assume there is a sequence $x^k$ such that $||x^k||_\infty\to\infty$ but $||T(x^k)||_\infty$ remains bounded. This implies especially the boundedness of $(x_0^k)$. But then again, $D_{t_1}^k=a_0^kF^{-1}(x_0^k)$ remains bounded. We can continue the argumentation for all coordinates of $T(x)$ and conclude that $(x_n^k)$ is a bounded sequence for all $n\in\{0,\dots,N\}$. This contradicts the assumption, and thus the lemma is proven.
\end{proof}
\begin{lemma}\label{lem:zero2}
Under the assumptions of Theorem \ref{thm:weiss2}, equation (\ref{eq:mainEq2}) has at most one solution, which is positive if existing. Furthermore, $g(x):=f(x)-\bar{a}_xf(\bar{a}_xx)(1-\tau\bar{\rho}'(x)x)$ is positive.
\end{lemma}
\begin{proof}
We show that both $h_2\circ F^{-1}$ and
\begin{equation}
\hat{h}_2(x):=-F^{-1}\left(X_0-N\left[x-F\left(\bar{a}_{F^{-1}(x)}F^{-1}(x)\right)\right]\right)
\end{equation}
are strictly increasing. In this case, at most one zero can exist, and its positivity is guaranteed by $h_2(F^{-1}(0))=0$ and $\hat{h}_2(0) = -F(X_0) < 0$. The function $h_2$ is strictly increasing because it is continuous, bijective, has a fixed point at zero and
\begin{eqnarray}
\lim_{\epsilon\to 0} \frac{h_2(\epsilon)-h_2(0)}{\epsilon}&=&\lim_{\epsilon\to 0}\frac{f(\epsilon)-\bar{a}_\epsilon^2f(\bar{a}_\epsilon \epsilon)(1-\tau\bar{\rho}'(\epsilon)\epsilon)}{f(\epsilon)-\bar{a}_\epsilon f(\bar{a}_\epsilon \epsilon)(1-\tau\bar{\rho}'(\epsilon)\epsilon)}\\
&=&\frac{1-\bar{a}_0^2}{1-\bar{a}_0} > 0.
\end{eqnarray}
Since $F^{-1}$ is also strictly increasing, we have proven the same property for $h_2\circ F^{-1}$. We differentiate $\hat{h}_2$:
\begin{equation}
\hat{h}_2'(x)=N\left[\frac{f(F^{-1}(x))-\bar{a}_{F^{-1}(x)}f(\bar{a}_{F^{-1}(x)}F^{-1}(x))(1-\tau\bar{\rho}'(F^{-1}(x))F^{-1}(x))}{f(F^{-1}(x))f\left(F^{-1}\left(X_0-N\left[x-F(\bar{a}_{F^{-1}(x)}F^{-1}(x))\right]\right)\right)}\right]\label{eq:hDiff2}
\end{equation}
This expression is strictly positive because the numerator is strictly positive as we show next. We define both
\begin{equation}
\begin{array}{ccl}
k(x)&:=&f(x)-\bar{a}_xf(\bar{a}_xx)(1-\tau\bar{\rho}'(x)x)\textrm{ and }\\
k_2(x)&:=&f(x)-\bar{a}^2_xf(\bar{a}_xx)(1-\tau\bar{\rho}'(x)x).
\end{array}
\end{equation}
The numerater of (\ref{eq:hDiff2}) can be expressed by $k(F^{-1}(x))$, and furthermore, $h_2(x)=xk_2(x)/k(x)$. Both functions $k$ and $k_2$ are continuous, and due to the properties of $h$ explained in the beginning of the proof, the functions must have the same sign for all $x\in\mathbb{R}$. The function $k_2$ is greater than $k$ for all $x\in\mathbb{R}$; thus, there can be no change of signs and we have either $k(x)>0$ and $k_2(x)>0$ or $k(x)<0$ and $k_2(x)<0$ for all $x$. Because $k(0)=f(0)(1-\bar{a}_0)>0$, positivity is proven.
\end{proof}
\begin{lemma}\label{lem:diff2}
For all $n\in\{0,\dots,N-1\}$, the partial derivatives of $C^{(1)}$ can be expressed by
\begin{eqnarray*}
\frac{\partial}{\partial x_n} C^{(2)}(x) &=& D_{t_n+}+\frac{a_nf(D_{t_{n+1}})(1-\tau\bar{\rho}'(D_{t_n+})D_{t_n+})}{f(D_{t_n+})}\left[\frac{\partial}{\partial x_{n+1}} C^{(2)}(x)-D_{t_{n+1}}\right].
\end{eqnarray*}
\end{lemma}
\begin{proof}
First, observe
\begin{equation}\label{eq:diffD}
\frac{\partial}{\partial x_n}D_{t_m}=\frac{a_nf(D_{t_{n+1}})}{f(D_{t_n+})}\left(1-\tau\bar{\rho}'(D_{t_n+})D_{t_n+}\right)\frac{\partial}{\partial x_{n+1}}D_{t_m}
\end{equation}\\
for $n\in\{0,\dots,m-2\}$. This follows from
\begin{eqnarray}
&&\frac{\partial}{\partial x_n}D_{t_m}\\
&=&\frac{\partial}{\partial x_n}\left[a_{m-1}F^{-1}(x_{m-1}+F(\dots(a_n F^{-1}(x_n+F(D_{t_n})))\dots))\right]\nonumber\\
&=&\left[\prod_{k=n+1}^{m-1}\left[\frac{d}{d x}\bar{a}_{F^{-1}(x_k+F(x))}F^{-1}(x_k+F(x))\right]_{x=D_{t_k}}\right]\left[\frac{\partial}{\partial x_n}\bar{a}_{F^{-1}(x_n+F(D_{t_n}))} F^{-1}(x_n+F(D_{t_n}))\right]\nonumber\\
&=&\left[f(D_{t_{n+1}})\frac{\partial}{\partial x_{n+1}}D_{t_m}\right]\left[\frac{a_n\left(1-\tau\bar{\rho}'(D_{t_n+})D_{t_n+}\right)}{f(D_{t_n+})}\right].\nonumber
\end{eqnarray}
We use (\ref{eq:diffD}) and (\ref{eq:costTransform2}) for the transformation
\begin{eqnarray}\label{eq:diffD2}
&&\frac{\partial}{\partial x_n}C^{(2)}(x)\\
&=&F^{-1}(x_n+F(D_{t_n}))+\sum_{m = n+1}^N\frac{\partial}{\partial x_n}\left[G(x_m+F(D_{t_m}))-\tilde{F}(D_{t_m})\right]\nonumber\\
&=&D_{t_n+}+\sum_{m = n+1}^N f(D_{t_m})\left[\frac{\partial}{\partial x_n}D_{t_m}\right]\left[F^{-1}(x_m+F(D_{t_m}))-D_{t_m}\right]\nonumber\\
&=&D_{t_n+}+\frac{a_nf(D_{t_{n+1}})}{f(D_{t_n+})}\left[1-\tau\bar{\rho} '(D_{t_n+})D_{t_n+}\right]\Big(D_{t_{n+1}+}-D_{t_{n+1}}\nonumber\\
&&+\sum_{m = n+2}^Nf(D_{t_m})\left[\frac{\partial}{\partial x_{n+1}}D_{t_m}\right]\left[F^{-1}(x_m+F(D_{t_m}))-D_{t_m}\right]\Big).\nonumber
\end{eqnarray}
Now, the same calculation for $\partial C^{(2)}(x) / \partial x_{n+1}$ results in
\begin{equation}\label{eq:diffD3}
\frac{\partial}{\partial x_n}D_{t_m} = D_{t_{n+1}+}+\sum_{m = n+2}^N f(D_{t_m})\left[\frac{\partial}{\partial x_{n+1}}D_{t_m}\right]\left[F^{-1}(x_m+F(D_{t_m}))-D_{t_m}\right],
\end{equation}
and combining (\ref{eq:diffD2}) and (\ref{eq:diffD3}) yields the desired result.
\end{proof}
Finally, we are prepared to prove Theorem \ref{thm:weiss2}.
\begin{lemma}
Strategy $\xi^{(2)}$ is the unique minimizer of function $C^{(2)}$ and all components of $\xi^{(2)}$ are positive.
\end{lemma}
\begin{proof}
Lemma \ref{lem:minExists2} guarantees the existence of at least one optimal strategy $\xi^*\in\Xi$. By standard arguments, there is a Lagrange multiplier $\nu\in\mathbb{R}$ such that
\begin{equation}
\frac{\partial}{\partial x^*_n}C^{(2)}(\xi^*)=\nu\textrm{\ \ \ for }n\in\{0,\dots,N\}.
\end{equation}
We use Lemma \ref{lem:diff2} to get
\begin{equation}
\nu=D_{t_n+}+\frac{a_nf(a_nD_{t_n+})}{f(D_{t_n+})}(1-\tau\bar{\rho}'(D_{t_n+})D_{t_n+})\left[\nu-a_nD_{t_n+}\right]
\end{equation}
\begin{equation}
\Leftrightarrow \nu = D_{t_n+}\frac{f(D_{t_n+})-a_n^2f(a_nD_{t_n+})(1-\tau\bar{\rho}'(D_{t_n+})D_{t_n+})}{f(D_{t_n+})-a_nf(a_nD_{t_n+})(1-\tau\bar{\rho}'(D_{t_n+})D_{t_n+})}=h_2(D_{t_n+})
\end{equation}
for $n\in\{0,\dots,N-1\}$. Function $h_2$ is one-to-one, and thus,
\begin{equation}
\nu=h_2(F^{-1}(x^*_n+F(D_{t_n})))
\end{equation}
implies that $x^*_n+F(D_{t_n})$ does not depend on $n\in\{0,\dots,N-1\}$. Consequently, $D_{t_n+}=F^{-1}(x^*_n+F(D_{t_n}))$ is constant in $n$ such that we can conclude
\begin{eqnarray}
x^*_0&=&F(h_2^{-1}(\nu)),\\
x^*_n&=&x^*_0-F(D_{t_n})=x^*_0-F(a_0F^{-1}(x^*_0))\textrm{ for }n\in\{1,\dots,N-1\}\label{eq:xnFromx02},\\
x^*_N&=&X_0-x^*_0-(N-1)\left[x^*_0-F(a_0F^{-1}(x^*_0))\right].
\end{eqnarray}
The value $x^*_0$ determines the optimal solution completely, and thus, it must minimize
\begin{eqnarray*}
&&C^{(2)}_0(x_0)\nonumber\\
&:=&C^{(2)}\left(x_0, x_0-F\left(a_0F^{-1}(x_0)\right),\dots,X_0-x_0-\left(N-1\right)\left[x_0-F(a_0F^{-1}(x_0))\right]\right)\nonumber\\
&\stackrel{(\ref{eq:costTransform2})}{=}&G(x_0)+\sum_{n=1}^{N-1}\left[G(x_0)-\tilde{F}(a_0F^{-1}(x_0))\right]\nonumber\\
&&\phantom{G(x_0)}+G(X_0-N[x_0-F(a_0F^{-1}(x_0))])-\tilde{F}(a_0F^{-1}(x_0))\nonumber\\
&=&N[G(x_0)-\tilde{F}(a_0F^{-1}(x_0))]+G(X_0-N[x_0-F(a_0F^{-1}(x_0))]).\nonumber
\end{eqnarray*}
\ \\
Differentiation results in
\begin{eqnarray}
\frac{dC^{(2)}_0(x_0)}{dx_0} &=&N\Bigg[D_{0+}-a_0^2D_{0+}\frac{f(D_{t_1})}{f(D_{0+})}\left(1-\tau\bar{\rho}'(D_{0+})D_{0+}\right)\\
&&\phantom{N\Bigg[}+D_{t_n+}\left(a_0\frac{f(D_{t_1})}{f(D_{0+})}\left(1-\tau\bar{\rho}'(D_{0+})D_{0+}\right)-1\right)\Bigg]\nonumber
\end{eqnarray}
such that we can calculate the minimizer by
\begin{equation}
\begin{array}{lcll}
&\frac{d}{dx^*_0}C^{(2)}_0(x^*_0)&=&0\label{eq:finEq}\\
\Leftrightarrow&D_{t_N+}&=&D_{0+}\frac{f(D_{0+})-a_0^2f(D_{t_1})(1-\tau\bar{\rho}'(D_{0+})D_{0+})}{f(D_{0+})-a_0f(D_{t_1})(1-\tau\bar{\rho}'(D_{0+})D_{0+})}.
\end{array}
\end{equation}
The left hand side of the last line can be rewritten as
\begin{eqnarray}
D_{t_N+}&=&F^{-1}(F(D_{t_N})+x^*_N)\\
&=&F^{-1}(F(D_{t_1})+X_0-x^*_0-(N-1)(x^*_0-F(D_{t_1})))\nonumber\\
&=&F^{-1}(X_0-N(x_0^*-F(D_{t_1}))\nonumber
\end{eqnarray}
and the right hand side is just $h_2(F^{-1}(x_0^*))$. We know by Lemma \ref{lem:zero2} that equation (\ref{eq:finEq}) has at most one zero such that we are finished with the existence, uniqueness and representation of the optimal strategy.

At last, we show that all components of this strategy are positive. We already know $x^*_0>0$ and thus also $x^*_n>0$ for all $n\in\{1,\dots,N-1\}$ by (\ref{eq:xnFromx02}). For the positivity of $x^*_N$, we transform (\ref{eq:finEq}) into
\begin{equation}
D_{t_N+}=D_{0+}\left[1+\frac{a_0f(a_0D_{0+})-a_0^2f(a_0D_{0+})}{f(D_{0+})-a_0f(a_0D_{t_{0+}})(1-\tau\bar{\rho}'(D_{0+})D_{0+})}\left(1-\tau\bar{\rho}'(D_{0+})D_{0+}\right)\right].
\end{equation}
The fraction on the right hand side is strictly positive by Lemma \ref{lem:zero2}; positivity of $x^*_N$ follows from
\begin{equation}
D_{t_N+}>D_{0+}=\frac{D_{t_N}}{a_0}>D_{t_N}.
\end{equation}
\end{proof}
\bibliographystyle{alpha}
\bibliography{weiss-liquidation}
\end{document}